\documentclass[aps,pra,superscriptaddress,longbibliography]{revtex4-2}

\usepackage{amsmath,amssymb,amsfonts,amsthm}
\usepackage{enumerate}
\theoremstyle{plain}
\newtheorem{theorem}{Theorem}[section]
\theoremstyle{remark}

\usepackage{graphicx}
\usepackage{hyperref}
\usepackage[english]{babel}

\begin{document}

\title{Gauged $SU(4)$ Strain Symmetry and the Non-Abelian Field Theory of Epidemiological Safe Zones: The Black Death Case}

\author{Jos\'e de Jes\'us Bernal-Alvarado}
\affiliation{Physics Engineering Department, Universidad de Guanajuato, M\'exico}
\author{David Delepine}
\affiliation{Physics Department, Universidad de Guanajuato, M\'exico}

\date{\today}

\begin{abstract}
Classical reaction-diffusion SIR models cannot explain why large geographic
regions -- such as the interior of Poland and Bohemia during the 14th-century
Black Death -- remained largely untouched while surrounding regions collapsed.
We argue that this anomaly is a natural consequence of a gauge-driven mechanism
acting on the four-lineage genetic structure of \emph{Yersinia pestis}.

Motivated by the paleogenomic ``Big Bang'' polytomy ($N=4$ branches), we
promote the strain index to an $SU(4)$ internal degree of freedom and
construct a Doi-Peliti field theory with a background spatial gauge field
$A_i(x)$ that encodes geographic and mobility heterogeneity. The local gauge
symmetry drives a Differential-Flow Turing-Hopf instability whose traveling
modes interfere to produce safe zones with a Bessel-function radial profile
$J_0(k_c r)$. We determine $k_c$ not from the observed safe-zone radius but
directly from the full $4\times4$ stability matrix, over a broad, randomized
sampling of the six symmetry-breaking cross-immunity couplings and two
diagonal reaction rates left unconstrained by 14th-century data, restricted
to draws satisfying the physical stability precondition of the underlying
theorem. Taking the diagonal gauge correlation length $\ell\sim1/A$ in the
range $25$--$30$~km -- independently consistent with documented medieval
carrier travel distances and market-town spacing -- the resulting
$R_{\rm safe}=2.4048/k_c$ is centered on $195$--$204$~km, in agreement with
the observed Poland-Bohemia anomaly ($150$--$200$~km), for $94$--$100\%$ of
the physically admissible parameter draws.

We further isolate the model's genuinely mutation-mediating channel, the
off-diagonal gauge field $A_i^{\rm off\text{-}diag}$, from reaction-sector
cross-immunity and diagonal advection, and derive the resulting mutation-wave
velocity in closed form. Above a threshold set by the intrinsic fitness
asymmetry between lineages, this velocity rises steeply with the off-diagonal
gauge amplitude, so that requiring agreement with the independently measured
paleogenomic lineage-radiation rate ($160$--$640$~km/year) fixes the
off-diagonal correlation length to $\ell_{\rm off}\simeq32$--$38$~km --
comparable to, and consistent with, the diagonal mobility scale, as expected
if both channels respond to the same regional environmental and mobility
heterogeneity.
\end{abstract}

\maketitle

\section{Introduction}
\label{sec:branches}
During the first pandemic wave of the Black Death (1346--1353), 
all of continental Europe suffered catastrophic mortality, with demographic
losses estimated at one-third to one-half of the affected population in
many regions~\cite{benedictow2004}. Against this background, the interior
of Poland and parts of Bohemia constitute a macroscopic anomaly: these
regions show substantially lower mortality, forming a large geographic
``safe zone'' surrounded by devastated territory~\cite{izdebski2022palaeoecological,christakos2005,olea2005,christakos2007}. The anomaly is not
explained by geographic isolation -- Poland was not bordered by impassable
barriers -- nor by documented quarantine measures comparable to those of
later pandemics. Benedictow~\cite{benedictow2004} documents the spatial
extent of this survival island at $150$--$200$ km in radius from the most
protected core.

Classical reaction-diffusion and metapopulation SIR models treat epidemic
waves as approximately homogeneous fronts of a single pathogen sweeping
across a landscape~\cite{noble1974geographic,keelingrohani}. Such models produce epidemic waves that
spread outward from an origin, but they generate no mechanism by which
a large contiguous region at the center of a spreading epidemic can remain
spared. The Poland-Bohemia anomaly thus poses a genuine explanatory
challenge to classical epidemic theory.

Paleogenomic analyses of ancient \emph{Yersinia pestis} DNA from medieval
burial sites have established that the Black Death was  caused by a rapid evolutionary radiation from
a common ancestral strain around 1346~CE~\cite{cui2013,spyrou2019}.
This radiation, sometimes called the ``Big Bang'' polytomy, produced at
least four major phylogenetic branches (Branches~1--4 of the $1.\text{PRE}$
clade) that diverged almost simultaneously and spread in different
geographic directions across Eurasia.

Ancient \emph{Y.~pestis} genomes confirmed as 14th-century Black Death
strains cluster near the root of this polytomy \cite{spyrou2019}. The polytomy structure is the key genomic
fact motivating the $N=4$ multiplet of our model.

Standard compartmental models (SIR, SEIR) describe the average behavior
of large populations via coupled ordinary differential equations~\cite{kermack1927,hethcote2000,keelingrohani}. But they discard all information about
demographic stochasticity. For spatial epidemic models, this
stochasticity matters: it generates fluctuations that can seed new outbreaks
and shapes the statistics of rare events such as long-distance
introductions~\cite{risken1989,strogatz2015}.

The Doi-Peliti formalism~\cite{doi1976,peliti1985,tauber2014critical,tauber2005,lefevre2007} allows us to generate path-integral field theories for stochastic
birth-death-type processes, recovering deterministic mean-field equations
as the saddle-point of the path integral and retaining fluctuation
corrections at higher orders. It was originally developed for population
dynamics and reaction-diffusion systems  as for instance the predator-prey (Lotka-Volterra)
stochastic master equation. In the Doi-Peliti
framework, each population compartment is represented by a bosonic
field, and the stochastic process is encoded in a Liouvillian operator
whose coherent-state path-integral representation defines a field theory
with a characteristic action $S_{\rm eff}=S_0+S_{\rm int}$~\cite{bernal2026gauge}. The free
kinetic action $S_0$ encodes transport (diffusion and advection), while
the interaction action $S_{\rm int}$ encodes local infection and recovery
events. We construct this action  for the four-lineage
\emph{Y.~pestis} system in Sec.~\ref{sec:IIB}.

The  paper is organized as follows.
Section~\ref{sec:IIA} develops the metapopulation SIR model for the $N=4$
strain multiplet and derives the form of the mobility matrix
$\mathbf D(x,t)$ from branch-conservation arguments.
Section~\ref{sec:IIB} constructs the $SU(4)$ Doi-Peliti field theory:
the master equation, the coherent-state path integral, the identification
of $SU(4)$ as the kinetic symmetry group, the gauging procedure, the
status of $A_i(x)$ as a background field, and the comparison with
classical metapopulation SIR.
Section~\ref{sec:III} proves the gauge-driven Turing-Hopf instability for
general $SU(N)$ via root-space decomposition (Theorem~\ref{thm:TH}),
verifies numerically that this instability persists in the full $N=4$
system over a broad sampling of the couplings left unconstrained by
14th-century data (Sec.~\ref{sec:fullmatrix}), and computes the critical
wavenumber $k_c$ for $SU(2)$ and $SU(4)$ as the wavenumber of maximum
growth (Sec.~\ref{sec:kc-su2-su4}). Section~\ref{sec:Rsafe-def} defines
the safe-zone radius from the general statistics of an
isotropically-selected critical wavenumber, valid at any $N$;
Sec.~\ref{sec:Rsafe-prediction} applies it to the Black Death, predicting
$R_{\rm safe}$ from the full stability matrix and an independent
historical estimate of the gauge correlation length;
Sec.~\ref{sec:IVD} derives the mutation-wave velocity from the
off-diagonal gauge sector and uses it to constrain that sector's
correlation length. Section~V is for discussions and Section VI is the
conclusions.

\section{From Metapopulation SIR to Gauged Doi-Peliti Field Theory}

\subsection{Metapopulation SIR and the continuum limit}
\label{sec:IIA}

Metapopulation SIR models are the standard tool for describing infectious
disease spread across a set of spatially distinct sub-populations (cities,
districts, or regions) coupled by host mobility, rather than a single
well-mixed population~\cite{keelingrohani,pastorsatorras2015}. They have been used to describe
the geographic spread of other epidemic diseases across networks of cities,
and, historically, the spread of plague along  commercial  and pilgrimage
routes. The essential ingredients are
\begin{enumerate}
    \item local SIR-type reaction kinetics
within each patch and
\item a mobility operator coupling patches, which in the
continuum limit becomes an effective spatial diffusion (for asymmetric
mobility, an advective bias)~\cite{edelstein2005,perthame2007}. 
\end{enumerate}
We adapt this model  generalizing it to the four-branch strain multiplet.

Let $S_i(t)$ denote the (strain-blind) susceptible population in patch $i$,
and let
\begin{equation}
\mathbf I_i(t) = \big(I_i^{(1)}(t),\dots,I_i^{(4)}(t)\big)^{\top}
\label{eq:Ii-def}
\end{equation}
denote the four-component vector of infected hosts in patch $i$, one entry
per branch of the Big Bang polytomy. Let $W_{ij}$ be the per-capita
mobility rate \emph{from} patch $j$ \emph{to} patch $i$ (the first index
is always the destination, the second the origin).

Pure migration in patch $i$ has two contributions: hosts \emph{arriving}
from every other patch $j$, at rate $W_{ij}S_j$, and hosts \emph{leaving}
patch $i$ toward every other patch $k$, at total per-capita rate
$\sum_{k\neq i}W_{ki}$.
Explicitly,
\begin{equation}
\left.\frac{dS_i}{dt}\right|_{\rm migration}
= \underbrace{\sum_{j\neq i} W_{ij}\,S_j}_{\text{inflow to } i}
\;-\;
\underbrace{\Big(\sum_{k\neq i} W_{ki}\Big) S_i}_{\text{outflow from } i}.
\label{eq:migration-inout}
\end{equation}
To write this compactly as a single sum $\sum_{j=1}^{M}L_{ij}S_j$
(including the $j=i$ term), the diagonal entry must reproduce the
outflow term, which fixes
\begin{equation}
L_{ij}=W_{ij}\ \ (i\neq j), \qquad
L_{ii}=-\sum_{k\neq i}W_{ki}
\quad\text{(total per-capita emigration rate out of patch $i$).}
\label{eq:Lij-def}
\end{equation}
With this definition, the metapopulation dynamics read
\begin{align}
\frac{dS_i}{dt} &= -S_i\, \boldsymbol\beta^\top \mathbf I_i
 + \sum_{j=1}^M L_{ij} S_j, \label{eq:metapop-S}\\
\frac{d\mathbf I_i}{dt} &= S_i\,\mathrm{diag}(\boldsymbol\beta)\,\mathbf I_i
 - \mathrm{diag}(\boldsymbol\gamma)\,\mathbf I_i
 + \sum_{j=1}^M L_{ij}\, \mathbf D_j\, \mathbf I_j,
\label{eq:metapop-I}
\end{align}
where we have separated the scalar spatial-coupling coefficient $L_{ij}$
from a $4\times4$ matrix $\mathbf D_j$ acting purely in strain space.

\paragraph{Structure of $\mathbf D$ from conservation of probability.}
Any such matrix can be written in the strain-space Lie-algebra basis as
\begin{equation}
\mathbf D = D_0\,\mathbb 1_4 + \sum_{a=1}^{15} d_a\, T^a,
\label{eq:Dmatrix}
\end{equation}
the identity plus the fifteen generators of $\mathfrak{su}(4)$ -- the unique
decomposition of a general Hermitian $4\times4$ matrix, with
$D_0=\tfrac14\mathrm{Tr}\,\mathbf D$ and $d_a=2\,\mathrm{Tr}(\mathbf D\,T^a)$.
Pure host mobility, however, is a strain-\emph{preserving} process: a host
carrying branch $(k)$ who moves from patch $j$ to patch $i$ still carries
branch $(k)$ on arrival. Conversion between branches is a mutation
(reaction) event, not a transport event, and belongs to $S_{\rm int}$ (or, in
the gauged theory of Sec.~\ref{sec:IIB}, to the off-diagonal components of
$A_\mu$) -- not to $\mathbf D$. Requiring that $\mathbf D$ commute with the
projector onto each branch (equivalently, that the total population of each
branch be separately conserved under $\mathbf D$ alone, in the absence of
reactions) implies  the coefficients $d_a$ of the twelve off-diagonal
(root-space) generators to be equal to zero:
\begin{equation}
\mathbf D = D_0\,\mathbb 1_4 + d_3 T^3 + d_8 T^8 + d_{15} T^{15}.
\label{eq:Dmatrix-cartan}
\end{equation}

\paragraph{Spatiotemporal dependence. }
 Geography (mountain ranges, river crossings, coastlines),
infrastructure (roads, waterways, later railways), and social structure
(trade networks, quarantine ordinances, market schedules) all vary from
place to place and evolve in time. So, one expects $\mathbf D$ to be spatially or temporally dependent: 
\begin{equation}
\mathbf D(x,t) = D_0(x,t)\,\mathbb 1_4 + d_3(x,t)\,T^3 + d_8(x,t)\,T^8 + d_{15}(x,t)\,T^{15}.
\label{eq:Dmatrix-xt}
\end{equation}
Treated purely within the classical metapopulation picture, $d_3(x,t)$,
$d_8(x,t)$, $d_{15}(x,t)$ would have to be introduced as three independent,
ad hoc position- and time-dependent functions, with no principle constraining
their form or their coupling to the reaction sector.
In the remainder of this paper we
adopt the simplest case $d_3=d_8=d_{15}=0$ for the residual classical
mobility matrix $\mathbf D$ itself, since no data
distinguishing branch-specific medieval mobility rates exist.

Embedding the patches on a lattice of spacing $\Delta x$ and assuming
nearest-neighbor mobility of characteristic length $\ell$ and rate $\lambda$,
it converges to a diffusion operator:
$$\sum_j L_{ij}\mathbf D\,\mathbf I_j \to \nabla\cdot\big[D_{\rm eff}(x,t)\nabla\mathbf I(x,t)\big]$$
with $D_{\rm eff}(x,t)\equiv D_0(x,t)\cdot\lambda\ell^2/2d$ in $d$ spatial
dimensions. An asymmetric
mobility matrix (prevailing trade or migration routes) additionally generates
an advective bias $v(x,t)\cdot\nabla\mathbf I$~\cite{marchetti2013}. Linearizing the infected
equation near a constant susceptible background $S_0$, for a single dominant
branch, and treating $D_{\rm eff}$ and $v$ as locally constant  gives the familiar scalar
advection-diffusion-reaction equation
\begin{equation}
\partial_t I = (\beta S_0-\gamma) I + D_{\rm eff}\nabla^2 I - v\cdot\nabla I,
\label{eq:fisherkpp}
\end{equation}
the standard Fisher-KPP form for an invading epidemic front~\cite{fisher1937,kolmogorov1937,murray2003mathematical,noble1974geographic}. 

\subsection{Doi-Peliti mapping and the $SU(4)$ strain multiplet}
\label{sec:IIB}

\subsubsection{Master equation and Fock-space representation}

At each lattice site $x$ we introduce bosonic creation/annihilation operators
for susceptible hosts, $a_{S,x}^\dagger, a_{S,x}$, and for infected hosts
carrying strain $i$,
\begin{equation}
a_{I,i,x}^\dagger,\ a_{I,i,x}, \qquad i=1,\ldots,4,
\end{equation}
satisfying canonical commutation relations $[a_{S,x},a_{S,y}^\dagger]=\delta_{xy}$,
$[a_{I,i,x},a_{I,j,y}^\dagger]=\delta_{ij}\delta_{xy}$. The state
$|P(t)\rangle=\sum_{\{n\}} P(\{n\},t)\,|\{n\}\rangle$  evolves as
$\partial_t|P(t)\rangle = -\hat{\mathcal L}|P(t)\rangle$, with Liouvillian
$\hat{\mathcal L}=\hat{\mathcal L}^{\rm trans}+\hat{\mathcal L}^{\rm reac}$.

The \emph{reaction} Liouvillian encodes the three elementary processes of the
multi-strain SIR model. Writing $a_S\equiv a_{S,x}$ etc.\ for brevity:
\begin{align}
\text{Infection (strain }i\text{)}: &\quad
S+I_i \xrightarrow{\beta_i} 2I_i, \nonumber\\
\hat{\mathcal L}_{\rm inf}^{(i)} &=
-\beta_i\bigl(a_{I,i}^\dagger a_{I,i}^\dagger - a_{S}^\dagger a_{I,i}^\dagger\bigr)
a_S\, a_{I,i}, \label{eq:Linf}\\[4pt]
\text{Recovery (strain }i\text{)}: &\quad
I_i \xrightarrow{\gamma_i} \varnothing, \nonumber\\
\hat{\mathcal L}_{\rm rec}^{(i)} &=
-\gamma_i\bigl(1 - a_{I,i}^\dagger\bigr)\,a_{I,i}. \label{eq:Lrec}
\end{align}
The \emph{transport} Liouvillian encodes nearest-neighbour hopping between
sites $x$ and $x+e_\mu$ with rates $\lambda_S$ and $\lambda_I$:
\begin{equation}
\hat{\mathcal L}^{\rm trans} = -\sum_\mu\sum_x\Bigl[
\lambda_S\bigl(a_{S,x+e_\mu}^\dagger - a_{S,x}^\dagger\bigr)a_{S,x}
+\sum_{i=1}^4\lambda_I\bigl(a_{I,i,x+e_\mu}^\dagger - a_{I,i,x}^\dagger\bigr)a_{I,i,x}
\Bigr].
\label{eq:Ltrans}
\end{equation}
Note that Eq.~\eqref{eq:Ltrans} uses a single hopping rate $\lambda_I$ common
to all four strains, in accordance with the uniform-mobility assumption
established in Sec.~\ref{sec:IIA}.

\subsubsection{Coherent-state path integral, $S_{\rm int}$, and recovery of SIR equations}

Introducing coherent states $|z_S,\{z_i\}\rangle$ with
$a_S|z_S\rangle=z_S|z_S\rangle$ and resolving the identity at each
infinitesimal time step, the generating functional of the stochastic process
becomes a path integral over complex fields
$\psi_S(x,t),\bar\psi_S(x,t)$ and $\psi(x,t)=(\psi_1,\ldots,\psi_4)^\top$,
$\bar\psi(x,t)=(\bar\psi_1,\ldots,\bar\psi_4)^\top$:
\begin{equation}
Z = \int \mathcal{D}[\bar\psi_S,\psi_S]\,\mathcal{D}[\bar\psi,\psi]\;
e^{-S_{\rm eff}[\bar\psi_S,\psi_S,\bar\psi,\psi]},
\end{equation}
with effective action $S_{\rm eff}=S_0+S_{\rm int}$. The kinetic (free) part
arising from Eq.~\eqref{eq:Ltrans} in the continuum limit reads
\begin{equation}
S_0[\psi_S,\bar\psi_S,\psi,\bar\psi] = \int d^dx\,dt\; \Big[
\bar\psi_S(\partial_t-D_S\nabla^2)\psi_S
+ \bar\psi^\top(\partial_t-D_I\nabla^2)\psi\Big],
\label{eq:S0}
\end{equation}
with $\bar\psi^\top\psi=\sum_{i=1}^4\bar\psi_i\psi_i$ and
$D_{S,I}=\lambda_{S,I}\ell^2/2d$.
The interaction part arising from Eqs.~\eqref{eq:Linf}--\eqref{eq:Lrec} is
\begin{equation}
S_{\rm int} = -\int d^dx\,dt\;\sum_{i=1}^4\Bigl[
\beta_i(\bar\psi_{I,i}-\bar\psi_S)\,\psi_S\,\psi_{I,i}
-\gamma_i\,\bar\psi_{I,i}\,\psi_{I,i}
\Bigr] + O(\bar\psi^2).
\label{eq:Sint}
\end{equation}
The $O(\bar\psi^2)$ terms encode stochastic fluctuations beyond the
mean-field level and are neglected in the saddle-point approximation.

\paragraph{Normalization and the physical manifold.}
The path integral satisfies $Z[0]=1$ (probability conservation), which in the
Doi-Peliti framework follows from the structure of the Liouvillian:
$\langle 1|\hat{\mathcal{L}}=0$, i.e., all reaction rates conserve total
probability. At the level of the action, this requires that $S_{\rm eff}$
vanishes on the \emph{physical manifold} $\bar\psi_S=\bar\psi_i=0$ for all
$i$: setting all response fields to zero must give $S_{\rm eff}=0$, which is
satisfied by both $S_0$ and $S_{\rm int}$ by inspection of
Eqs.~\eqref{eq:S0}--\eqref{eq:Sint}.

\paragraph{Mean-field equations and recovery of metapopulation SIR.}
The saddle-point equations $\delta S_{\rm eff}/\delta\bar\psi_S=0$ and
$\delta S_{\rm eff}/\delta\bar\psi_{I,i}=0$, evaluated on the physical
manifold and identifying $\psi_S\equiv S(x,t)$ and $\psi_{I,i}\equiv I_i(x,t)$
with macroscopic host densities, give
\begin{align}
\partial_t S &= D_S\nabla^2 S - \sum_{i=1}^4\beta_i\, S\, I_i,
\label{eq:mfS}\\
\partial_t I_i &= D_I\nabla^2 I_i + \beta_i\, S\, I_i - \gamma_i\, I_i,
\qquad i=1,\ldots,4.
\label{eq:mfI}
\end{align}
Equations~\eqref{eq:mfS}--\eqref{eq:mfI} are precisely the continuum
metapopulation SIR equations of Sec.~\ref{sec:IIA} (with $\mathbf D=D_I\mathbb
1_4$), confirming that $S_{\rm int}$ is the correct interaction sector. The
full stochastic theory encoded in $S_{\rm eff}$ is a systematic extension
of these mean-field equations to include demographic noise.

\subsubsection{Global symmetry of $S_0$: $U(4)$ vs.\ $SU(4)$ and the role of $D_0$}
\label{sec:IIB3}

The symmetry group of the free action $S_0$ depends  on whether the
diffusion coefficient $D_I$ is the \emph{same scalar} for all four strains.
Consider the general case in which each branch has its own diffusion
coefficient: $S_0\to\int d^dx\,dt\; \bar\psi^\top(\partial_t-\hat D\nabla^2)\psi$,
with $\hat D=\mathrm{diag}(D_1,D_2,D_3,D_4)$. Under a global $U\in U(4)$,
$S_0\mapsto\int d^dx\,dt\;\bar\psi^\top U^\dagger(\partial_t-\hat D\nabla^2)U\psi$,
which reduces to $S_0$ only if $U^\dagger\hat D\,U=\hat D$, i.e.\ if $U$
lies in the commutant of $\hat D$ in $U(4)$:
\begin{equation}
\begin{cases}
\hat D = D_0\,\mathbb 1_4 \;\Rightarrow\; G_{\rm sym}=U(4),\\[4pt]
\hat D = \mathrm{diag}(D_1,D_2,D_3,D_4),\ D_i\text{ distinct}\;\Rightarrow\;
G_{\rm sym}=U(1)^4.
\end{cases}
\label{eq:symbreaking}
\end{equation}
The full $U(4)$ symmetry thus requires that \emph{all four branches diffuse
with the same coefficient} $D_0$. Decomposing $U(4)\cong SU(4)\times U(1)$,
the $U(1)$ factor is the uniform phase rotation $\psi_i\mapsto e^{i\alpha}\psi_i$,
which generates the conserved total infected charge
$Q=\int d^dx\sum_i|\psi_i|^2$ (fixed by normalization, carrying no
internal information). The non-trivial internal symmetry mixing the four
branches is therefore
\begin{equation}
G_{\rm internal} = U(4)/U(1) \cong SU(4).
\end{equation}

In the present paper we adopt the uniform-diffusion assumption
$\hat D=D_0\mathbb 1_4$ for the following reason: no 14th-century data exist
that would allow distinguishing the mobility rates of individual \emph{Y.\
pestis} lineages -- paleogenomic sampling is far too sparse to reconstruct
branch-specific geographic spread velocities. 

\subsubsection{Local $SU(4)$ gauge symmetry and $A_i(x)$}
 
\emph{Spatial gauge field only.} The Doi-Peliti action~\eqref{eq:S0} is
first-order in time and second-order in spatial coordinates -- a
non-relativistic structure inherited directly from the master equation.
The most general locally $SU(4)$-invariant kinetic operator that
preserves this structure is
\begin{equation}
D_t - D_I\,\mathbf D^2
= (\partial_t + A_t) - D_I\sum_{i=1}^d(\partial_i+A_i)^2,
\label{eq:fullcov}
\end{equation}
with the full set of gauge-field components $A_t(x,t),A_i(x,t)\in
\mathfrak{su}(4)$ and covariant derivative $D_\mu=\partial_\mu+A_\mu$.
Unlike a unitary quantum theory, the Doi-Peliti fields $\psi$ and
$\bar\psi$ transform independently under the gauge group,
\begin{equation}
\psi \mapsto U\psi, \qquad \bar\psi \mapsto \bar\psi\,U^{-1}, \qquad
U(x,t)\in SU(4),
\end{equation}
rather than being related by Hermitian conjugation as in a unitary
theory. Covariance of the bilinear pairing $\bar\psi D_\mu\psi$ under
this transformation, $D_\mu\psi\mapsto U(D_\mu\psi)$, fixes the
transformation law of the gauge field purely algebraically,
\begin{equation}
A_\mu \mapsto U A_\mu U^\dagger - (\partial_\mu U)U^\dagger,
\label{eq:gaugetrans}
\end{equation}
with no Hermiticity condition imposed on $A_\mu$ itself. 

So, we have to comment
some important notes to understand well the model.
\begin{enumerate}
    \item \textbf{ The temporal gauge $A_t=0$.}
Starting from Eq.~\eqref{eq:fullcov} with $A_t\neq0$, a local gauge
transformation $U(x,t)=\mathcal T\exp\!\bigl(-\int_0^t A_t(x,t')\,dt'\bigr)$
(time-ordered exponential) transforms $A_t\to 0$ while shifting
$A_i\to A_i^U$ to a new gauge-equivalent spatial field.
Setting $A_t=0$ is therefore a \emph{gauge choice} (the temporal gauge),
not a restriction on the time-dependence of the theory.  In this gauge
Eq.~\eqref{eq:fullcov} reduces to $\partial_t - D_I\mathbf D^2$.

\item \textbf{The non-relativistic kinetic structure.}
Unlike the relativistic d'Alembertian $\partial_\mu\partial^\mu$, the
Doi-Peliti kinetic operator is asymmetric in space and time
($\partial_t$ versus $D_I\nabla^2$). This asymmetry means that the
gauge coupling of the time derivative ($A_t$) and the spatial
derivatives ($A_i$) carry different physical dimensions and different
epidemiological interpretations.
 
\end{enumerate}

In this paper we work in the temporal gauge $A_t=0$ so the gauge field reduces to
\begin{equation}
A_i(x,t)\in\mathfrak{su}(4), \quad i=1,\ldots,d,
\qquad A_t=0.
\label{eq:Aspatial}
\end{equation}
 
\paragraph{Physical meaning of $A_t$.}
To make the gauge choice transparent, we note what a non-zero $A_t(x,t)$
would represent: a position- and time-dependent strain-selective pressure.
Its coupling $\bar\psi^\top A_t\psi = A_t^a\,\bar\psi^\top T^a\psi$ acts
as a local chemical potential differentially shifting the growth rate of
each lineage at each point in time -- for instance, seasonal fluctuations
in flea-vector abundance, time-varying host immunity, or spatially targeted
quarantine measures.  In our model these effects are already encoded in
the time-dependent kinetic parameters $\beta_i(t)$, $\gamma_i(t)$ within
$S_{\rm int}$ (Eq.~\eqref{eq:Sint}); working in the temporal gauge $A_t=0$
is therefore physically consistent: the time-varying component of
environmental selective pressure is absorbed into the reaction sector, not
the gauge sector.  
 
In the temporal gauge~\eqref{eq:Aspatial}, the relevant spatial covariant derivative and its Laplacian are
\begin{equation}
\mathbf D \equiv \nabla + A(x),
\qquad
\mathbf D^2 \equiv \sum_{i=1}^d D_iD_i
= \nabla^2 + 2A_i(x)\partial_i + A_i(x)A_i(x).
\label{eq:covD}
\end{equation}
Under a local $SU(4)$ transformation $U(x)\in SU(4)$, the gauge field
transforms as $A_i(x)\mapsto U A_i U^\dagger - (\partial_i U)U^\dagger$,
ensuring $\mathbf D\psi\mapsto U\mathbf D\psi$. Promoting $\nabla^2\to\mathbf D^2$
in $S_0$ gives the locally $SU(4)$-invariant kinetic action (Sec.~\ref{sec:IIB5}).
 
The $15$ generators of $\mathfrak{su}(4)$ split, in the Cartan basis, into
$3$ diagonal and $12$ off-diagonal generators; correspondingly,
$A_i(x)=A_i^{\rm diag}(x)+A_i^{\rm off\text{-}diag}(x)$:
\begin{itemize}
\item \textbf{Diagonal components} $A_i^{\rm diag}(x)$ act as
strain-differentiated drift velocities, $v_k(x)=2D_I[A_i^{\rm diag}(x)]_{kk}$
-- directional environmental/mobility gradients that favour or disfavour
particular branches at each point in space.
Ordinary diffusion (set by $\mathbf D$) and
gauge-induced differential advection (set by $A_i^{\rm diag}$) play
structurally distinct roles and enter additively in $\mathbf D^2$ as  $A_i^{\rm diag}(x)$ lies in the Cartan subalgebra of $\mathfrak{su}(4)$ and is therefore \emph{traceless}.
 
\item \textbf{Off-diagonal components} $A_i^{\rm off\text{-}diag}(x)$
rotate the multiplet between branches as infected hosts disperse across
environmental gradients . This is the genuinely non-Abelian content of the
model, with no classical metapopulation SIR counterpart.
 
\emph{Data limitation.} The off-diagonal sector cannot be quantitatively
constrained with available 14th-century data, which would require
time-resolved geographic sampling of individual lineages. We restrict the
quantitative analysis to the diagonal sector $A_i^{\rm diag}(x)$; the
off-diagonal sector is retained structurally as a necessary consequence of
the $SU(4)$ construction and discussed qualitatively in Sec.~V.
\end{itemize}

\subsubsection{The gauged free action $S_0^{\rm gauged}$ and its normalization}
\label{sec:IIB5}

Replacing $\nabla^2\to\mathbf D^2$ (Eq.~\eqref{eq:covD}) in the kinetic
action~\eqref{eq:S0} gives the locally $SU(4)$-invariant gauged action
\begin{equation}
S_0^{\rm gauged} = \int d^dx\,dt\;\Bigl[
\bar\psi_S(\partial_t - D_S\nabla^2)\psi_S
+ \bar\psi^\top\bigl(\partial_t - D_I\,\mathbf D^2\bigr)\psi
\Bigr].
\label{eq:S0gauged}
\end{equation}
Note that only the spatial Laplacian is promoted to its covariant form, in
accordance with the non-relativistic structure established in
Sec.~\ref{sec:IIB}. Expanding $\mathbf D^2$:
\begin{equation}
D_I\,\mathbf D^2\,\psi = D_I\nabla^2\psi
+ 2D_I\, A_i(x)\,\partial_i\psi + D_I\, A_i(x)A_i(x)\,\psi,
\label{eq:covlap}
\end{equation}
where repeated spatial index $i$ is summed.  The first term is ordinary diffusion (already in
the ungauged $S_0$); the second is the gauge-induced differential advection,
responsible for the Turing-Hopf instability of Sec.~\ref{sec:III}; the third
is a ``mass'' correction shifting the  growth rates of different branches,
important when $|A_i|\ell\gtrsim 1$.

\paragraph{Normalization.}
The normalization $Z[0]=1$ of the ungauged theory is preserved: $A_i(x)$ enters
only through the covariant spatial Laplacian. 

\paragraph{Mean-field equations with gauge field.}
Taking the saddle-point of $S_0^{\rm gauged}+S_{\rm int}$ gives the mean-field
equations
\begin{align}
\partial_t S &= D_S\nabla^2 S - \sum_i\beta_i\,S\,I_i,
\label{eq:mfSg}\\
\partial_t \mathbf I &= D_I\nabla^2\mathbf I
- 2D_I\, A_i(x)\,\partial_i\mathbf I
- D_I\, A_i(x)A_i(x)\,\mathbf I + \mathbf R(S,\mathbf I),
\label{eq:mfIg}
\end{align}
where $\mathbf R(S,\mathbf I)=\mathrm{diag}(\beta_k S-\gamma_k)\,\mathbf I$
collects the local reaction terms. Setting $A_i=0$ recovers
Eqs.~\eqref{eq:mfS}--\eqref{eq:mfI} exactly. Equation~\eqref{eq:mfIg} is
the starting point for the linear stability analysis of Sec.~\ref{sec:III}.


\emph{c. Status of $A_\mu(x)$: background field, not a mediating particle.}
As emphasized in Sec.~\ref{sec:branches}, $A_\mu(x)$ here is \emph{not}
promoted to a dynamical field with its own Yang-Mills kinetic term. We
instead treat it as a \emph{quenched} background field, borrowing the
term from the physics of disordered systems (e.g., spin glasses): a
quenched field is one that is frozen relative to the dynamics of the
variables it couples to, $\partial_t A_i=0$, rather than an
\emph{annealed} field that would itself relax dynamically and have to
be integrated over in the path integral. This is the appropriate
treatment whenever the field varies on a timescale much longer than
that of the process it influences.

This separation of timescales is precisely what holds here. The
characteristic epidemic timescale $\tau_{\rm epi}$ -- set by the
infection and recovery rates $\beta_i,\gamma_i$ -- and the
characteristic timescale $\tau_{\rm env}$ over which the underlying
mobility and geographic structure changes (new roads, shifting trade
networks, changing settlement patterns) operate on very different
scales:
\begin{equation}
\tau_{\rm epi}\sim 1\text{--}10\ \text{days}
\;\ll\;
\tau_{\rm env}\sim 1\text{--}10\ \text{years}.
\end{equation}
Because $A_i(x,t)$ is determined by exactly this slowly-varying
environmental and mobility structure, the quenched approximation
$\partial_t A_i=0$ is well justified over the duration of a local
epidemic wave. Table~\ref{tab:XI} makes this comparison explicit.
\begin{table}[t]
\caption{Comparison between the present background-field $SU(4)$ model and a
standard dynamical Yang-Mills gauge theory.}
\label{tab:XI}
\begin{tabular}{lll}
\hline\hline
Property & This model & Yang-Mills \\
\hline
Gauge group & $SU(4)$ & $SU(N)$ \\
Gauge symmetry & Local, exact on $S_0$ & Local, exact \\
Gauge field $A_\mu$ & Fixed background & Dynamical \\
Determined by & Mobility/environment data & Field equations \\
Yang-Mills term & Absent & Present \\
Gauge bosons & None (quenched) & Propagating \\
Covariant derivative & $D_i=\partial_i+A_i$ & $D_\mu=\partial_\mu+A_\mu$ \\
Coupling constraint & Minimal coupling only & Minimal coupling only \\
\hline\hline
\end{tabular}
\end{table}

\subsection{Comparison: classical metapopulation SIR vs.\ gauged field theory}
\label{sec:IIC}

 We 
make explicit, point by point, what the gauged $SU(4)$ Doi-Peliti theory adds
to the classical metapopulation SIR model of Sec.~\ref{sec:IIA}, and why each
addition is either necessary or well-motivated given the available data.
Table~\ref{tab:comparison} provides a summary.

\paragraph{ Spatial transport: the diagonal sector is equivalent.}
The diagonal components $A_\mu^{\rm diag}(x,t)$ of the gauge field act,
in the mean-field equations~\eqref{eq:mfIg}, as an $(x,t)$-dependent
advective drift $v_i(x,t)=2D_I[A_\mu^{\rm diag}(x,t)]_{ii}$, different for
each branch $i$. Restricted to this sector, Eq.~\eqref{eq:mfIg} reduces
to a set of decoupled, branch-specific advection-diffusion-reaction equations
of the Fisher-KPP type~\eqref{eq:fisherkpp}. In this limit the gauged theory
is \emph{classically equivalent} to a metapopulation SIR model with
branch-dependent, spatially heterogeneous mobility.

Restricted to the diagonal sector alone, the gauge construction is
therefore \emph{equivalent} to this ad hoc picture: the
tracelessness of $A_\mu^{\rm diag}$ only enforces $\sum_k v_k=0$,
removing the physically unobservable common-mode drift (which can
always be absorbed into a comoving frame), but otherwise
$A_\mu^{\rm diag}(x,t)$ remains an arbitrary, data-unconstrained
function of space and time.

b. \emph{Cross-branch coupling: the symmetric sector is fixed by gauge
invariance.} In the classical multi-strain SIR model, coupling between
branches can be introduced via cross-immunity or competitive-exclusion
terms $\alpha_{ij}I_iI_j$ added to the reaction term $R(S,\mathbf I)$
of Eq.~\eqref{eq:mfIg}. With $N=4$ branches, the general coupling
matrix $\alpha_{ij}$ contains $N^2=16$ independent entries; decomposing
it into symmetric and antisymmetric parts,
$\alpha_{ij}=\alpha_{ij}^{\rm sym}+\alpha_{ij}^{\rm asym}$
($\alpha_{ij}^{\rm asym}=-\alpha_{ji}^{\rm asym}$), each part is
constrained differently by the gauge symmetry.

Imposing $SU(4)$ invariance on the symmetric part forces it to take
the form
\begin{equation}
\alpha_{ij}^{\rm sym} = \alpha_0\,\delta_{ij}
+ \alpha_1\sum_{a=1}^{15}(T^a)_{ij}(T^a)_{ji},
\label{eq:alpha-sym}
\end{equation}
the unique $SU(4)$-invariant quadratic combination of generators (the
quadratic Casimir bilinear in the fundamental representation). Using
the $SU(N)$ Fierz identity \cite{vanritbergen1999}, $\sum_a(T^a)_{ij}(T^a)_{ji}=\tfrac12$ for
every pair $i\neq j$ is a fixed constant. Eq.~\eqref{eq:alpha-sym} reduces the symmetric sector from
$O(N^2)$ free parameters to exactly two real numbers
$(\alpha_0,\alpha_1)$, regardless of $N$. This is the direct analogue
of the Yukawa coupling structure in gauge theories of particle
physics~\cite{weinberg1996,nakahara2003}: the symmetry principle
selects, among all symmetric interaction forms, the minimal one
consistent with the assumed internal symmetry.

No analogous reduction applies to the antisymmetric part: the Fierz
sum above is manifestly symmetric under $i\leftrightarrow j$, so no
$SU(4)$-invariant antisymmetric bilinear of this type exists. The
antisymmetric sector,
\begin{equation}
\alpha_{ij}^{\rm asym} = \kappa_{ij}, \qquad \kappa_{ij}=-\kappa_{ji},
\label{eq:alpha-asym}
\end{equation}
therefore enters as an explicit, gauge-symmetry-breaking ingredient of
the model, with six independent components $\kappa_{ij}$
($1\le i<j\le4$) not fixed by $SU(4)$ covariance. Physically,
$\kappa_{ij}$ encodes genuine strain-pair-specific immunological
asymmetry -- e.g., prior exposure to branch $i$ conferring stronger
cross-protection against $j$ than the reverse -- a pattern documented
in multi-strain pathogen dynamics more broadly~\cite{gog2002,
andreasen1997}. This asymmetry is not expected to respect the kinetic
$SU(4)$ symmetry. Adding both pieces to the reaction term of
Eq.~\eqref{eq:mfIg} gives the full cross-coupled reaction term,
\begin{equation}
R_p(S,\mathbf I) = (\beta_pS-\gamma_p)I_p
- \sum_{q\neq p}\big(\alpha_{pq}^{\rm sym}+\kappa_{pq}\big) I_pI_q.
\label{eq:R-extended}
\end{equation}

\paragraph{ Off-diagonal sector: mutation mediated by geography.}
The off-diagonal components $A_\mu^{\rm off\text{-}diag}(x,t)$ generate
genuine non-Abelian mixing: they rotate the four-component field $\psi$
between different branches as hosts move through space. This has no
counterpart in any classical multi-strain SIR model and is the mechanism
through which environmental gradients (climate, host density, vector
abundance) can drive the transition from one \emph{Y.\ pestis} lineage to
another during geographic dispersal. As stated in Sec.~\ref{sec:IIB}, we do
not quantitatively constrain this sector with medieval data; its role in the
present paper is structural (it is a necessary consequence of the gauged
$SU(4)$ construction) and qualitative.It provides a theoretical mechanism
for the observed phylogenetic diversity of the Black Death strains.

\paragraph{ Emergence of the Turing-Hopf mechanism.}
The most consequential difference between the classical and gauged models is
not in the equilibrium structure but in the \emph{stability} of the
homogeneous endemic state. In the classical metapopulation SIR model (even
with multiple strains and branch-dependent mobilities in the Cartan sector),
spatial homogeneity can be destabilized only through the classical
Turing mechanism -- which requires the diffusion coefficients of different
species to be sufficiently different (activator-inhibitor structure)~\cite{turing1952chemical,segel1972dissipative,cross2009pattern,hoyle2006pattern}. In the
single-effective-branch Fisher-KPP limit, no spatial instability occurs at
all: the front simply propagates.

The gauged $SU(4)$ theory introduces a qualitatively new mechanism: the
gauge-induced advection term $-2D_IA_\mu^{\rm diag}\cdot\nabla\mathbf I$ in
Eq.~\eqref{eq:mfIg} is \emph{branch-differential} (since $A_\mu^{\rm diag}$
is traceless, different branches are advected in different directions) even
when the diffusion coefficient $D_I$ is the same for all branches. This
differential flow can trigger a \emph{Differential-Flow Turing-Hopf}
instability -- distinct from the classical Turing mechanism -- which
generates traveling waves of mutation even when the homogeneous state is
locally (at $k=0$) stable~\cite{rovinsky1992,cross1993}. This instability is the subject of Sec.~\ref{sec:III},
and its application to the Black Death safe zones is in Sec.~IV. It is the
central physical mechanism of the paper and has no classical analogue.
\paragraph{ Why is $SU(4)$ preferred over other groups?}
Classical metapopulation SIR imposes no symmetry on the strain-interaction
matrix $\alpha_{ij}$ -- any real matrix is a priori allowed. Once one decides
to impose a compact Lie group symmetry to reduce the parameter space, the
natural candidates are $SU(N)$, $SO(N)$, and $Sp(N)$. 
 As shown in
Sec.~\ref{sec:IIB3}, $SU(4)$ is the unique maximal compact symmetry of
the kinetic action $S_0$ under the conditions (i) $D_0$ uniform across
branches and (ii) $U(1)$ factored out. This is a \emph{conditional} result:
if one accepts conditions (i) and (ii) -- both justified above -- $SU(4)$
follows necessarily. Alternative groups $SO(N)$ or $Sp(N)$ would require
additional conditions that have no natural justification in the SIR context.

\begin{table}[t]
\caption{Comparison between classical metapopulation SIR and the gauged
$SU(4)$ Doi-Peliti field theory developed in this section.}
\label{tab:comparison}
\begin{tabular}{p{3.8cm}p{3.0cm}p{3.8cm}}
\hline\hline
Ingredient & Classical SIR & Gauged $SU(4)$ \\
\hline
Spatial diffusion & $D_{\rm eff}$, scalar & $D_I\mathbb 1_4$, uniform \\
Mobility heterogeneity & ad hoc $v(x,t)$ & $A_\mu^{\rm diag}(x,t)$, principled \\
Cross-branch coupling & $O(N^2)$ free params & 2 invariants
$(\alpha_0,\alpha_1)$ + 6 symmetry-breaking $\kappa_{pq}$ \\
Mutation via geography & absent & $A_\mu^{\rm off\text{-}diag}$ (structural) \\
Spatiotemporal structure of couplings & unconstrained & fixed by gauge covariance \\
Turing-Hopf mechanism & absent & present (differential flow) \\
Constrained by data here & $D_{\rm eff}$, $\beta$, $\gamma$ & $D_I$, $\beta_i$, $\gamma_i$, $A_\mu^{\rm diag}$ \\
Not constrained (medieval) & --- & $A_\mu^{\rm off\text{-}diag}$, $d_{3,8,15}$ \\
\hline\hline
\end{tabular}
\end{table}

\section{Gauge-Driven Turing-Hopf Instability for General $SU(N)$}
\label{sec:III}

We analyse the linear stability of the spatially homogeneous endemic
state under the gauged mean-field equations~\eqref{eq:mfSg}--\eqref{eq:mfIg}.
The strategy is to treat the general $SU(N)$ case directly using root-space
algebra.

\subsection{Setup and notation for general $SU(N)$}
\label{sec:IIIA}

Let $(\tilde S_0, \tilde{\mathbf I}_0)$ be a spatially homogeneous endemic
steady state of the local reaction kinetics,
$\mathbf R(\tilde S_0,\tilde{\mathbf I}_0)=\mathbf 0$. We consider small
perturbations $\delta S(x,t)$ and
$\delta\mathbf I(x,t)=(\delta I_1,\ldots,\delta I_N)^\top$.
Linearising Eq.~\eqref{eq:mfIg} and specialising to a one-dimensional
geometry with $A_x = A\,H$, $H=\mathrm{diag}(h_1,\ldots,h_N)\in\mathfrak{su}(N)$
a Cartan element (diagonal, traceless: $\sum_i h_i=0$), gives
\begin{equation}
\partial_t\,\delta\mathbf I
= \mathbf J\,\delta\mathbf I
+ D_I\nabla^2\delta\mathbf I
- 2D_I A H\,\partial_x\delta\mathbf I,
\label{eq:linearised}
\end{equation}
where $\mathbf J = \partial\mathbf R/\partial\mathbf I\big|_*$ is the
$N\times N$ Jacobian at the steady state. The plane-wave ansatz
$\delta\mathbf I = \mathbf v\,e^{\lambda t+ikx}$ gives
\begin{equation}
\mathbf M(k)\,\mathbf v = \lambda\,\mathbf v,
\qquad
\mathbf M(k) = \mathbf J - D_Ik^2\mathbb 1_N - 2iD_IkA\,H.
\label{eq:Mk}
\end{equation}
The homogeneous state is stable iff $\mathrm{Re}[\lambda_\ell(k)]<0$ for
all $\ell$ and all $k\in\mathbb R$. By assumption $\mathbf J$ is stable at
$k=0$; the question is whether spatial structure combined with $A\neq 0$
can drive $\mathrm{Re}[\lambda]>0$.

\subsection{Root-space decomposition for general $SU(N)$}
\label{sec:IIIB}

For $N\geq 3$ the characteristic polynomial of $\mathbf M(k)$ is degree $N$
and intractable in general. We exploit the Lie-algebraic structure of
$\mathfrak{su}(N)$.

\paragraph{Root system of $SU(N)$.}
The algebra $\mathfrak{su}(N)$ has rank $N-1$ and root system $A_{N-1}$ \cite{georgi1999, humphreys1972}.
The positive roots in the fundamental representation are
$\alpha_{ij}=e_i-e_j$ ($1\leq i<j\leq N$), with corresponding root space
spanned by the elementary matrix $E_{ij}$, $(E_{ij})_{ab}=\delta_{ia}\delta_{jb}$.
Each root $\alpha_{ij}$ selects a two-dimensional root subspace
$\{E_{ij},E_{ji}\}$.

\paragraph{Action of $H$ on root subspaces.}
For the Cartan element $H=\mathrm{diag}(h_1,\ldots,h_N)$,
$H\cdot\alpha_{ij}=h_i-h_j$. The gauge-induced advection speed for a mode
in the $(i,j)$ root subspace is $\phi_{ij}(k)=2D_IkA(h_i-h_j)$.

\paragraph{Block decomposition of $\mathbf M(k)$.}
In the eigenbasis of $H$,
$[\mathbf M(k)]_{ij}=J_{ij}-(D_Ik^2+2iD_IkAh_i)\delta_{ij}$.
For a general root pair $(i,j)$ with $J_{ij}\neq0$, the
$(i,j)$-principal $2\times2$ submatrix of $\mathbf M(k)$ is exactly
\begin{equation}
\mathbf M^{(ij)}(k)
= \bar\tau_{ij}(k)\,\mathbb{1}
+\begin{pmatrix}
\delta_{ij}-i\phi_{ij}(k) & J_{ij}\\
J_{ji} & -\delta_{ij}+i\phi_{ij}(k)
\end{pmatrix},
\label{eq:Mroot-general}
\end{equation}
with
\begin{align}
\bar\tau_{ij}(k) &= \frac{J_{ii}+J_{jj}}{2}-D_Ik^2 - iAD_Ik\,(h_i+h_j),
\label{eq:tau-bar}\\
\delta_{ij} &= \frac{J_{ii}-J_{jj}}{2},\\
\phi_{ij}(k) &= D_IkA\,(h_i-h_j).
\label{eq:phi-corrected-B}
\end{align}
The common shift $\bar\tau_{ij}(k)$ is real, and Eq.~\eqref{eq:Mroot-general}
reduces to the simpler antipodal form:
\begin{equation}
\mathbf M^{(ij)}(k)
= \begin{pmatrix}
J_{ii}-D_Ik^2-i\phi_{ij}(k) & J_{ij}\\
J_{ji} & J_{jj}-D_Ik^2+i\phi_{ij}(k)
\end{pmatrix},
\label{eq:Mroot}
\end{equation}
whenever $h_i+h_j=0$. This holds in particular for the dominant pair
$(p,q)=(1,N)$ since $h_1+h_N=0$ by construction of the symmetric Cartan element $H_N$,and, for $N=4$, also for the pair $(2,3)$. For the remaining,
non-antipodal root pairs of the $N=4$ system ($(1,2),(1,3),(2,4),(3,4)$, the
full complex $\bar\tau_{ij}(k)$ of Eq.~\eqref{eq:tau-bar} must be retained;
$\mathrm{Re}[\lambda_\pm^{(ij)}(k)]$ is then given by
$$\mathrm{Re}[\bar\tau_{ij}(k)]
\pm\mathrm{Re}[\sqrt{D_r(k)+iD_i(k)}]$$
with $D_r,D_i$ as in
Eqs.~\eqref{eq:Dr-Di-general}, unaffected by the imaginary part of
$\bar\tau_{ij}$.

\subsection{Main theorem: $SU(N)$ gauge-driven Turing-Hopf instability}
\label{sec:IIIC}

\begin{theorem}[$SU(N)$ Gauge-Driven Turing-Hopf Instability]
\label{thm:TH}
Let $(\tilde S_0,\tilde{\mathbf I}_0)$ be a stable homogeneous endemic
steady state (all eigenvalues of $\mathbf J$ have strictly negative real
parts). Let $H=\mathrm{diag}(h_1,\ldots,h_N)\in\mathfrak{su}(N)$ have
distinct entries. Suppose there exist $p,q\in\{1,\ldots,N\}$ such that:
\begin{enumerate}[(C1)]
\item $J_{pp}\neq J_{qq}$ \emph{(strain asymmetry)};
\item $J_{pq}\neq 0$ and $J_{qp}\neq 0$ \emph{(non-zero cross-coupling)};
\item $J_{pq}J_{qp}<0$ and $|J_{pq}J_{qp}|>\delta_{pq}^2$,
$\delta_{pq}=(J_{pp}-J_{qq})/2$
\emph{(competitive asymmetric cross-immunity)}.
\end{enumerate}
Then for any $A>A_{\rm crit}(p,q)$, where
\begin{equation}
A_{\rm crit}(p,q) := \frac{1}{|h_p-h_q|}
\sqrt{\frac{\left| \delta_{pq}^2+J_{pq}J_{qp} \right|}{4D_I^2 k_\star^2}},
\label{eq:Acrit}
\end{equation}
and $k_\star$ minimises $|\tau_{pq}(k)|=(J_{pp}+J_{qq})/2-D_Ik^2$,
there exists $k_c\neq 0$ such that $\mathbf M(k_c)$ has at least one
eigenvalue with $\mathrm{Re}[\lambda(k_c)]>0$: a Turing-Hopf
traveling-wave instability is driven by the gauge field.
\end{theorem}

\begin{proof}[Proof sketch]
The $2\times 2$ block $\mathbf M^{(pq)}(k)$ (Eq.~\ref{eq:Mroot}) satisfies
conditions (C1)--(C3) by hypothesis. Its discriminant
$D^{(pq)}(k)=\delta_{pq}^2+J_{pq}J_{qp}-4D_I^2k^2A^2(h_p-h_q)^2
-2iD_IkA(h_p-h_q)(J_{pp}-J_{qq})$
has $D^{(pq)}_r(0)=\delta_{pq}^2+J_{pq}J_{qp}<0$ by (C3), so the
$(p,q)$ subspace is already a stable focus at $k=0$. For $k>0$, $D^{(pq)}_r$
decreases further while $D^{(pq)}_i\neq 0$ by (C1). For
$A>A_{\rm crit}(p,q)$ [Eq.~\eqref{eq:Acrit}] the instability condition is
satisfied at $k=k_c$. The unstable eigenvalue of
$\mathbf M^{(pq)}$ persists in the full $N\times N$ system by continuity.
\end{proof}

Conditions (C1) and (C2) together are equivalent to $[H,\mathbf J]\neq\mathbf 0$
on the $(p,q)$ subspace -- the gauge field and Jacobian must not commute.

\emph{Origin of condition (C3) in the interaction sector.} Recalling
the reaction term of Eq.~\eqref{eq:R-extended}, $J_{pq}=-(\alpha_{pq}^{\rm sym}
+\kappa_{pq})I_p^*$ for $p\neq q$. The $SU(4)$-invariant symmetric
piece alone gives $J_{pq}J_{qp}=(\alpha_{pq}^{\rm sym})^2I_p^*I_q^*
\geq0$ identically -- the wrong sign for condition (C3), for any
choice of $(\alpha_0,\alpha_1)$. Condition (C3) can therefore only be
satisfied through the symmetry-breaking piece $\kappa_{pq}$,
confirming that the competitive, asymmetric cross-immunity driving the
gauge-induced instability is a genuine additional ingredient of the
model, not a consequence of $SU(4)$ covariance.

\subsubsection{Dominant root.}
Since $A_{\rm crit}(p,q)\propto1/|h_p-h_q|$ at fixed
$(\delta_{pq},J_{pq}J_{qp},k_\star)$ (Eq.~\eqref{eq:Acrit}), and $A$
itself is common to all root pairs, the pair with the \emph{largest}
Cartan separation $|h_p-h_q|$ has the \emph{smallest} threshold and
is therefore the first to destabilize as $A$ increases from zero. We
call this the dominant root, and define its separation as
\begin{equation}
\Phi_{\max}(H):=\max_{i<j}|h_i-h_j|,
\label{eq:Phimax-def}
\end{equation}
the largest pairwise gap between eigenvalues of a given Cartan element
$H$.

Throughout this paper we adopt the canonical, evenly-spaced Cartan
element
\begin{equation}
H_N=\mathrm{diag}(N-1,N-3,\ldots,-(N-1)),
\label{eq:HN-def}
\end{equation}
i.e. $h_i=N+1-2i$, with unit spacing $h_i-h_{i+1}=2$ between
consecutive entries -- the natural choice representing a generic
linear ordering of the $N$ branches with no preferred internal
structure beyond that ordering (Sec.~\ref{sec:IIIA}). Because the
entries of $H_N$ are strictly decreasing, the largest gap trivially
occurs between the first and last entries,
\begin{equation}
\Phi_{\max}(H_N) = h_1-h_N  = 2(N-1),
\label{eq:Phimax-HN}
\end{equation}
so the dominant root is the extrema pair $(p,q)=(1,N)$. 
\subsubsection{Numerical verification of the continuity argument, and
dependence on the gauge amplitude $A$}
\label{sec:fullmatrix}

That the unstable eigenvalue of the isolated $2\times2$ block
$\mathbf M^{(pq)}(k)$ persists in the full $N\times N$ system ``by
continuity'' -- is not automatic: it requires the coupling between the
dominant root subspace $(p,q)$ and the remaining branches to be small
compared to the instability margin. We verify this explicitly for $N=4$,
and, since the gauge amplitude $A$ was left unconstrained by the
formalism up to this point, we determine how the critical wavenumber
$k_c$ of the full system depends on it.

\paragraph{Construction of the full Jacobian.}
We build the full $4\times4$ Jacobian consistent with
Eq.~\eqref{eq:R-extended}. The diagonal entries $J_{11},J_{44}$ and the
symmetry-breaking coupling $\kappa_{14}=-(J_{14}+\alpha_{\rm sym})$
 are fixed
by $(\delta_J,J_{14}J_{41})$. We set
$\alpha_{\rm sym}=0$: since $\kappa_{14}^2=\alpha_{\rm sym}^2-J_{14}J_{41}$
and $J_{14}J_{41}<0$ is fixed, $|\kappa_{14}|$ is minimized at
$\alpha_{\rm sym}=0$, which in turn minimizes the overall coupling scale
entering the stability problem and gives the most conservative (least
fine-tuned) test of persistence.  We reat $J_{22},J_{33}$ as free
parameters, sampled independently and uniformly on $[J_{44},J_{11}]$,
together with the five remaining symmetry-breaking couplings $\kappa_{pq}$
(sampled independently in sign and magnitude, uniform in
$[1,3]\times\kappa_{14}$), since none of these seven quantities is
constrained by 14th-century data.

\paragraph{Stability precondition.}
Theorem~\ref{thm:TH} requires $(\tilde S_0,\tilde{\mathbf I}_0)$ to be a
stable homogeneous endemic state, i.e.\ all eigenvalues of $\mathbf J$
(the reaction Jacobian alone, with no gauge field and no diffusion) must
have strictly negative real part. We enforce this explicitly: for each
random draw of $(J_{22},J_{33},\kappa_{12},\kappa_{13},\kappa_{23},
\kappa_{24},\kappa_{34})$ we diagonalize $\mathbf J$ at $A=0$, $k=0$, and
discard the draw if $\mathrm{Re}[\lambda_{\max}(\mathbf J)]\geq0$. Over
$10^5$ draws, the acceptance rate is $18.2\%$; all results below refer
only to the accepted (physically admissible) draws.

\paragraph{Critical wavenumber of the full system.}
For each accepted draw, 
we define, consistently with Eq.~\eqref{eq:kc-def}, the wavenumber of
maximum growth of the \emph{full} system,
\begin{equation}
k_c^{\rm full} \;\equiv\; \underset{k>0}{\arg\max}\;
\mathrm{Re}\big[\lambda_{\max}(\mathbf M(k))\big],
\qquad
\mathbf M(k)=\mathbf J-D_Ik^2\mathbb 1-2iD_IkA\mathbf H,
\label{eq:kc-full-def}
\end{equation}
obtained numerically for each draw. Since $\mathbf M(k)$ depends on the
gauge amplitude $A$, so does $k_c^{\rm full}$.

\paragraph{Dependence on $A$.}
We take $A\simeq 1/\ell$, with $\ell$ a correlation length of the
environmental and mobility heterogeneity encoded in the background gauge
field, and repeat the $10^5$-draw sampling
above for a range of $\ell$. Table~\ref{tab:kc-vs-ell} reports, for each
$\ell$:the fraction of accepted draws that
develop a gauge-driven instability,
$\mathrm{Re}[\lambda_{\max}(\mathbf M(k_c^{\rm full}))]>0$; and the
resulting distribution of $k_c^{\rm full}$.

\begin{table}[htbp]
\centering
\begin{tabular}{cccccccc}
\hline
$\ell$ (km) & $A=1/\ell$ (km$^{-1}$) & unstable\% &
$k_c^{\rm full}$ p5 & p25 & p50 & p75 & p95 \\
& & & \multicolumn{5}{c}{(km$^{-1}$, $\times10^{-3}$)}\\
\hline
20  & 0.0500 & 100.0 & 9.25  & 10.29 & 11.06 & 11.91 & 13.13 \\
25  & 0.0400 & 100.0 & 9.66  & 10.90 & 11.80 & 12.82 & 14.33 \\
30  & 0.0333 & 99.9  & 9.93  & 11.29 & 12.31 & 13.49 & 15.28 \\
40  & 0.0250 & 94.7  & 6.27  & 11.33 & 12.46 & 13.80 & 15.84 \\
50  & 0.0200 & 84.6  & 5.38  & 10.19 & 11.53 & 12.95 & 15.05 \\
75  & 0.0133 & 58.7  & 5.07  & 6.89  & 8.50  & 10.36 & 12.73 \\
100 & 0.0100 & 42.9  & 4.63  & 6.43  & 7.50  & 8.59  & 10.76 \\
\hline
\end{tabular}
\caption{Dependence of the full-system critical wavenumber
$k_c^{\rm full}$ [Eq.~\eqref{eq:kc-full-def}] on the gauge correlation
length $\ell$ ($A=1/\ell$), over $10^5$ random draws of
$(J_{22},J_{33},\kappa_{pq})$ per $\ell$, restricted to draws satisfying
the Theorem~\ref{thm:TH} stability precondition (acceptance $18.2\%$,
independent of $\ell$). Percentiles are computed only over draws that
develop a gauge-driven instability at their own $k_c^{\rm full}$.}
\label{tab:kc-vs-ell}
\end{table}

For $\ell\lesssim30$~km, essentially all admissible draws (99.9--100\%)
develop a gauge-driven instability, and $k_c^{\rm full}$ is tightly
distributed around $(1.1$--$1.2)\times10^{-2}~\mathrm{km^{-1}}$
irrespective of the random values of $(J_{22},J_{33},\kappa_{pq})$. As
$\ell$ increases beyond $\sim40$~km, both the fraction of unstable draws
and the typical $k_c^{\rm full}$ decrease steadily. 
\subsection{$k_c$ in $SU(2)$ and $SU(4)$}
\label{sec:kc-su2-su4}

Theorem~\ref{thm:TH} establishes that a gauge-driven Turing-Hopf instability
exists whenever $A>A_{\rm crit}(p,q)$, but it does not by itself fix the
wavenumber of the fastest-growing mode. Since it is precisely this wavenumber
that sets the safe-zone length scale in Sec.~IV
(via $R_{\rm safe}=2.4048/k_c$), we compute $k_c$  for the two cases relevant to
this paper: $SU(2)$, used as an analytically transparent illustration, and
$SU(4)$, the physical model applied to the Black Death.

\subsubsection{General setup: $k_c$ as the wavenumber of maximum growth}
\label{sec:kc-general}

For the dominant root pair $(p,q)$, the two eigenvalues of the $2\times2$
block $M^{(pq)}(k)$ of Eq.~\eqref{eq:Mroot} are
\begin{equation}
\lambda_\pm(k) \;=\; \tau(k) \,\pm\, \sqrt{D_r(k)+iD_i(k)},
\qquad
\tau(k) = \frac{J_{pp}+J_{qq}}{2} - D_Ik^2,
\label{eq:lambda-pm-general}
\end{equation}
with
\begin{equation}
\varphi_{pq}(k) = D_I A\,(h_p-h_q)\,k,
\qquad
D_r(k) = \delta_{pq}^2 + J_{pq}J_{qp} - \varphi_{pq}(k)^2,
\qquad
D_i(k) = -2\,\varphi_{pq}(k)\,\delta_{pq},
\label{eq:Dr-Di-general}
\end{equation}
and $\delta_{pq}=(J_{pp}-J_{qq})/2$ as before. We define the critical
wavenumber as the wavenumber of the fastest-growing unstable mode,
\begin{equation}
k_c(p,q) \;\equiv\; \underset{k>0}{\arg\max}\;\mathrm{Re}[\lambda_+(k)],
\qquad
\mathrm{Re}[\lambda_+(k)] = \tau(k) + \sqrt{\dfrac{\sqrt{D_r(k)^2+D_i(k)^2}+D_r(k)}{2}}.
\label{eq:kc-def}
\end{equation}
Setting $d\,\mathrm{Re}[\lambda_+]/dk=0$ in Eq.~\eqref{eq:kc-def} gives a
condition of the form
\begin{equation}
D_I = \frac{1}{4\big(\mathrm{Re}[\lambda_+(k)]-\tau(k)\big)}\,
\frac{d}{dk}\!\left[\sqrt{D_r(k)^2+D_i(k)^2}+D_r(k)\right].
\label{eq:extremum-condition}
\end{equation}
 Equation
\eqref{eq:kc-def} is  transcendental in $k$, and $k_c(p,q)$ must in
general be obtained numerically, by direct maximization of
$\mathrm{Re}[\lambda_+(k)]$. 

\subsubsection{$k_c$ for $SU(2)$}
\label{sec:kc-su2}

For $N=2$, $H=\mathrm{diag}(1,-1)\equiv\sigma_3$, so $h_1-h_2=2$ and
Eq.~\eqref{eq:Dr-Di-general} gives $\varphi_{12}(k)=2D_IAk$.
\begin{equation}
M(k) =
\begin{pmatrix}
J_{11}-D_Ik^2-2iD_IkA & J_{12}\\[2pt]
J_{21} & J_{22}-D_Ik^2+2iD_IkA
\end{pmatrix},
\label{eq:M-SU2-explicit}
\end{equation}
with $D_r(k)=\delta_J^2+J_{12}J_{21}-4D_I^2k^2A^2$ and
$D_i(k)=-4D_IkA\,\delta_J$ [our Eq.~\eqref{eq:Dr-Di-general}], $\delta_J=(J_{11}-J_{22})/2$.
This is the unique case in the present paper for which the instability
condition, Eq.~\eqref{eq:Acrit}, can be written out fully in closed algebraic form; even
so, that expression is an inequality bounding $A$ at fixed $k$, not a
solution for the maximizing $k_c$ itself, and the extremum problem
$d\,\mathrm{Re}[\lambda_+]/dk=0$ retains the nested-radical structure of
Eq.~\eqref{eq:extremum-condition}.

For the illustrative parameters of Fig.~(\ref{fig:SU2})
($J_{11}=0.7$, $J_{22}=-0.9$, $J_{12}=-J_{21}=1$, $D_I=0.1$, $A_{\rm
eff}=2$, all in dimensionless units), numerically maximizing
Eq.~\eqref{eq:kc-def} gives
\begin{equation}
k_c^{SU(2)} = 1.1075, \qquad \mathrm{Re}[\lambda_+(k_c)] = 0.1925,
\label{eq:kc-su2-numeric}
\end{equation}
in agreement with the peak of the dispersion curve shown in Fig.~(\ref{fig:SU2})(A) and
used to construct the safe-zone profile of Fig.~(\ref{fig:SU2})(B)--(C).

Figure~\ref{fig:SU2} illustrates the three key aspects of the $SU(2)$
gauge-driven instability using dimensionless parameters
($J_{11}=0.7$, $J_{22}=-0.9$, $J_{12}=-J_{21}=1$, $D_I=0.1$,
$A_{\rm eff}=2$) that satisfy conditions (C1)--(C3) and place the system
well above the instability threshold.

Panel~(A) shows the dispersion relation $\mathrm{Re}[\lambda_+(k)]$: at $A=0$
the homogeneous endemic state is stable for all $k$ (blue curve stays negative);
switching on $A=A_{\rm eff}$ drives $\mathrm{Re}[\lambda_+]$ positive in a band
around $k_c\approx 1.11$ (red curve), confirming the gauge field as the
necessary trigger for the instability.

Panel~(B) shows the resulting radial safe-zone profile $S(r)\propto J_0(k_cr)$, derived in
Sec.~\ref{sec:Rsafe-def} from the general statistics of an
isotropically-selected critical wavenumber:
the central susceptible region (high $S$, green shading) is given by the
first root of $J_0$ at $r=R_{\rm safe}=2.4048/k_c\approx 2.17$ (normalized
units), beyond which the infection field takes over.

Panel~(C) shows the two-dimensional version: the safe zone appears as a roughly
circular high-$S$ patch surrounded by high-infection territory. The dashed circle
marks $R_{\rm safe}$. This topology -- a Bessel-function nodal disc -- is the
structural prediction of the gauge-driven Turing-Hopf mechanism for any $SU(N)$
model in two spatial dimensions, valid at any $N$ (Sec.~\ref{sec:Rsafe-def}),
and its quantitative application to the Black Death ($SU(4)$, medieval
mobility parameters) is in Sec.~\ref{sec:Rsafe-prediction}.

\begin{figure}[h]
\centering
\includegraphics[width=\linewidth]{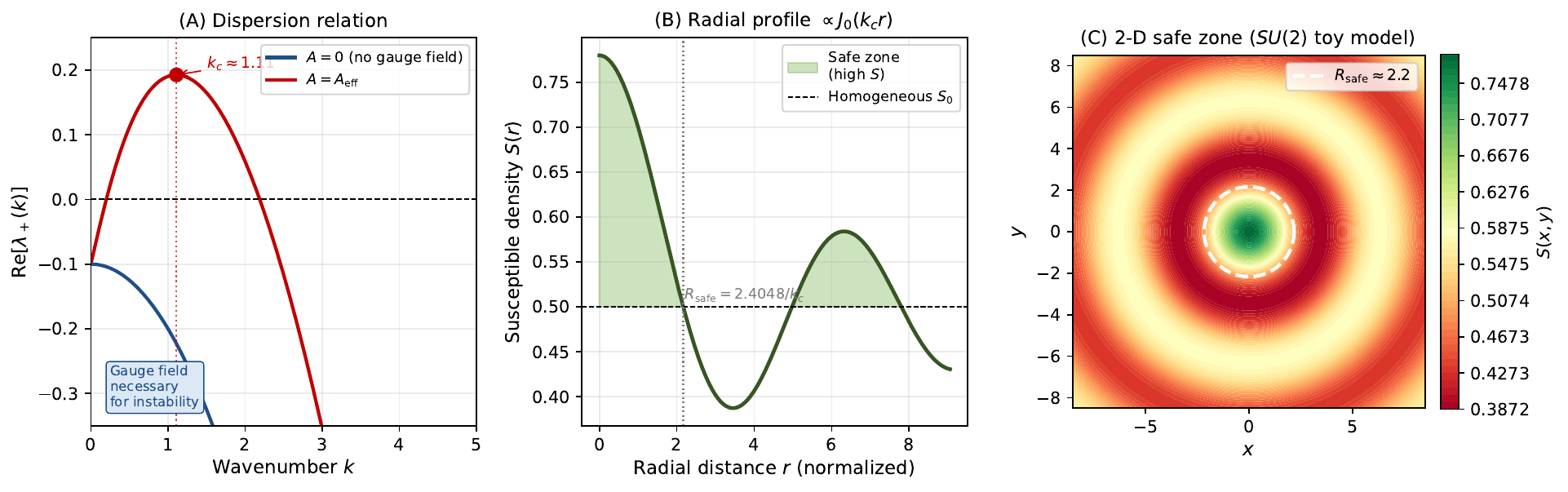}
\caption{$SU(2)$ gauge-driven Turing-Hopf instability and safe-zone structure
(dimensionless toy model). \textbf{(A)}~Dispersion relation
$\mathrm{Re}[\lambda_+(k)]$: stable for all $k$ without the gauge field
(blue); unstable in a band around $k_c\approx1.11$ with $A=A_{\rm eff}$ (red).
\textbf{(B)}~Radial susceptible-density profile $S(r)\propto J_0(k_cr)$:
the safe zone (green shading, high $S$) extends to the first zero of $J_0$
at $r=R_{\rm safe}\simeq2.4048/k_c$. \textbf{(C)}~Two-dimensional safe-zone
map; the dashed circle marks $R_{\rm safe}$. The $J_0$ topology is a generic
prediction of the $SU(N)$ gauge mechanism; the corresponding quantitative
application to the Black Death uses $N=4$ and medieval mobility parameters
(Sec.~\ref{sec:Rsafe-prediction}).}
\label{fig:SU2}
\end{figure}

\subsubsection{$k_c$ for $SU(4)$}
\label{sec:kc-su4}

For $N=4$ we use the canonical Cartan element $H_4=\mathrm{diag}(3,1,-1,-3)$
of Eq.~\eqref{eq:HN-def}, for which the dominant root pair is the extremal one,
$(p,q)=(1,4)$, with separation $h_1-h_4=2(N-1)=6$
[Eq.~\eqref{eq:Phimax-HN}]. From Eq.~\eqref{eq:Dr-Di-general},
\begin{equation}
\varphi_{14}(k) = 6\,D_IA\,k,
\qquad
D_r(k) = \delta_J^2+J_{14}J_{41}-36D_I^2A^2k^2,
\qquad
D_i(k) = -12\,D_IAk\,\delta_J,
\label{eq:Dr-Di-SU4}
\end{equation}
with $\delta_J=(J_{11}-J_{44})/2$ as in Sec.~III.C. As for the general case,
the extremum condition on Eq.~\eqref{eq:kc-def} reduces, after two
successive squarings, to an octic polynomial in $k$ whose roots are not
analytically illuminating; we therefore solve
$\arg\max_k\mathrm{Re}[\lambda_+(k)]$ numerically, using the same
maximization procedure validated against the $SU(2)$ case above.

Here, we
adopt representative strain-asymmetry values satisfying conditions
(C1)--(C3) of Theorem~\ref{thm:TH}, $\delta_J=0.046~\mathrm{day^{-1}}$
and $J_{14}J_{41}=-2.26\times10^{-3}~\mathrm{day^{-2}}$ (i.e.\ a
competitive, asymmetric cross-immunity coupling
$|\kappa_{14}|>|\delta_J|$, as required by (C3)), together with
$D_I=D_{\rm eff}=70~\mathrm{km^2/day}$ (Sec.~\ref{sec:IVC}) and a bare
gauge amplitude $A=A_{\rm bare}=0.0011~\mathrm{km^{-1}}$. The numerical
maximization of Eq.~\eqref{eq:kc-def} then gives
\begin{equation}
k_c^{SU(4)} = 6.861\times10^{-3}~\mathrm{km^{-1}},
\qquad
\mathrm{Re}[\lambda_+(k_c)] = 5.074\times10^{-3}~\mathrm{day^{-1}}.
\label{eq:kc-su4-numeric}
\end{equation}
This is the dynamically critical wavenumber of the gauge-driven
Turing-Hopf instability for the $(1,4)$ root pair, obtained  from
diagonalizing $M^{(14)}(k)$.

\section{Application to the Black Death}

The Turing-Hopf instability of Sec.~\ref{sec:III} establishes that the
gauge field drives traveling waves of mutation in the infection field.
In this section, we characterize the resulting spatial patterns quantitatively and compare
them with the Poland-Bohemia anomaly.

\subsection{Definition of the safe-zone radius}
\label{sec:Rsafe-def}

First we have to define what is the meaning of safe-zone radius $R_{\rm safe}$ in our case.

\paragraph{Origin of the Bessel profile at finite $N$.}
Near the Turing-Hopf threshold, growth is confined to a narrow band of
wavenumbers around the critical magnitude $k_c$ established in
Sec.~\ref{sec:fullmatrix}; modes with $|\mathbf k|\neq k_c$ decay. Crucially,
$\mathrm{Re}[\lambda_{\max}(\mathbf k)]$, as computed from
Eq.~\eqref{eq:kc-full-def}, depends only on $|\mathbf k|$: no direction in
the physical $(x,y)$ plane is privileged by the construction of
Sec.~\ref{sec:IIIA}, since the gauge field's orientation in strain space
is assumed isotropic.  Under
 random initial conditions,the growing field is therefore well
described as an incoherent superposition of plane waves with the same
magnitude $k_c$ for the wave number but random orientation and phase,
\begin{equation}
\delta S(\mathbf x) \;\approx\; \int_0^{2\pi}\frac{d\theta}{2\pi}\,
a(\theta)\,e^{ik_c\,\mathbf n(\theta)\cdot\mathbf x} + \text{c.c.},
\qquad \mathbf n(\theta)=(\cos\theta,\sin\theta),
\label{eq:random-superposition}
\end{equation}
with amplitudes $a(\theta)$ uncorrelated across $\theta$
($\langle a(\theta)a^*(\theta')\rangle\propto\delta(\theta-\theta')$ for
genuinely random initial conditions). The two-point correlation of such a
field is, using the standard integral representation
$J_0(z)=\frac{1}{2\pi}\int_0^{2\pi}e^{iz\cos\alpha}\,d\alpha$,
\begin{equation}
\big\langle \delta S(\mathbf x)\,\delta S(\mathbf x+\mathbf r)\big\rangle
\;\propto\;
\int_0^{2\pi}\frac{d\theta}{2\pi}\,e^{ik_c\,\mathbf n(\theta)\cdot\mathbf r}
\;=\; J_0(k_c r),
\qquad r=|\mathbf r|,
\label{eq:bessel-derivation}
\end{equation}
so that the radial envelope around any local extremum of the pattern
(such as the safe-zone nodes visible in Fig.~(\ref{fig:validation})A) follows a Bessel profile
with the same $k_c$ that sets the growth rate. This is a general property
of any two-dimensional, isotropically-selected pattern-forming
instability with incoherent excitation of the critical ring
$|\mathbf k|=k_c$ \cite{cross1993}; it holds for any $N$, including the
physical case $N=4$ treated here.

\paragraph{Definition.}
We therefore define the safe-zone radius as the first zero of the profile
derived in Eq.~\eqref{eq:bessel-derivation},
\begin{equation}
R_{\rm safe} \;:=\; \frac{2.4048}{k_c},
\label{eq:Rsafe-definition}
\end{equation}
where $k_c$ is, throughout the remainder of Sec.~IV, the $N=4$ critical
wavenumber $k_c^{\rm full}(A)$ of Eq.~\eqref{eq:kc-full-def}.

\begin{figure}[htbp]
\centering
\includegraphics[width=0.85\columnwidth]{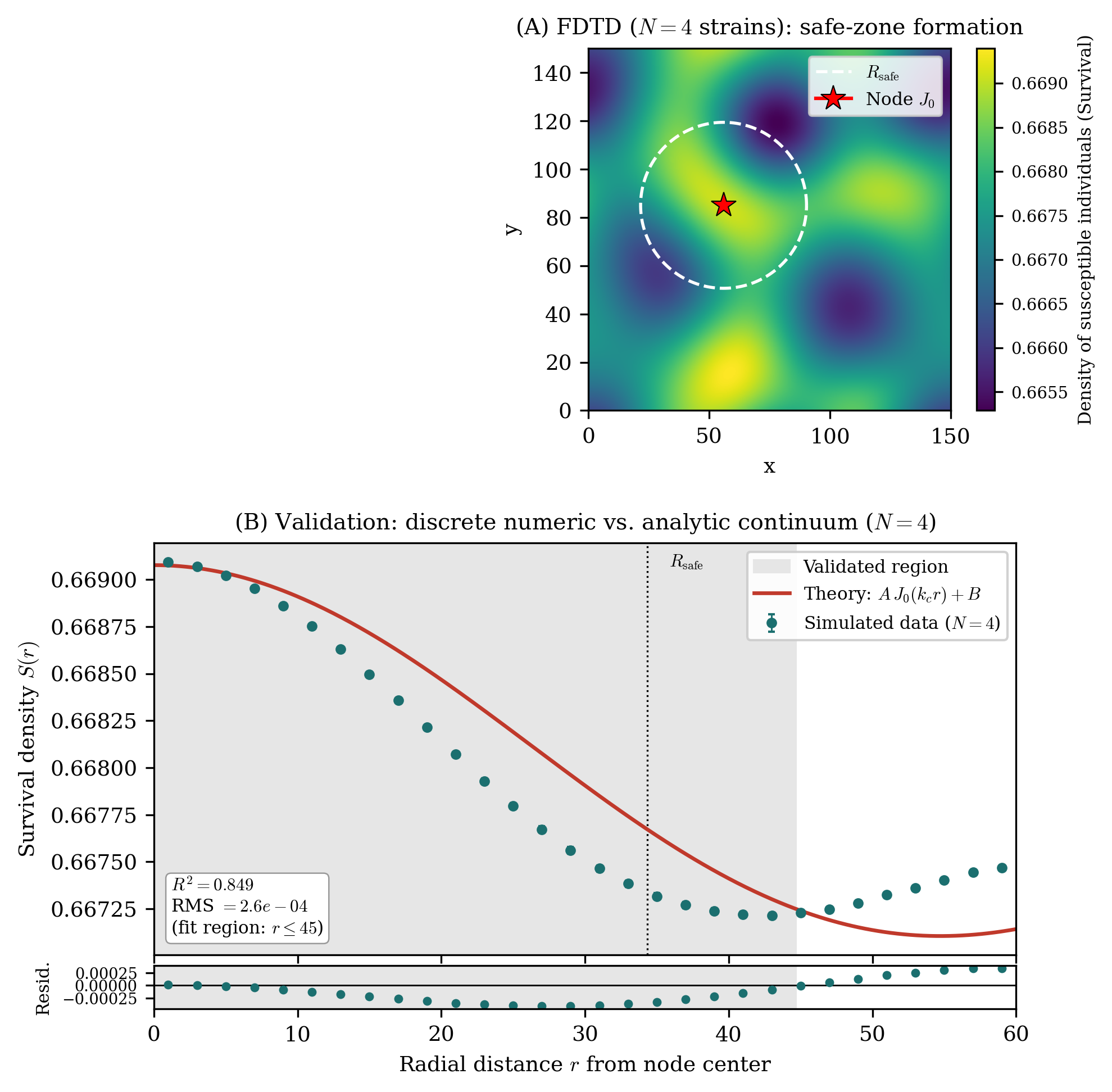}
\caption{Validation of the gauged $SU(4)$ Doi-Peliti model against a
direct numerical simulation ($N=4$ strains). \textbf{(A)} Susceptible
density $S(x,y)$ from a finite-difference time-domain (FDTD)
realization of the linearized mean-field equations
[Eqs.~\ref{eq:mfSg}--\ref{eq:mfIg}] with random
initial conditions: constructive interference of the four
gauge-advected strain waves produces localized high-$S$ safe zones
embedded in a disordered Turing-Hopf pattern. The star marks the most
prominent safe-zone node; the dashed circle marks the predicted
radius $R_{\rm safe}=2.4048/k_c$. \textbf{(B)} Radially binned
susceptible density $S(r)$ about that node (teal points, mean $\pm$
S.E.M.\ per annular bin) compared with the isotropic-pattern
prediction $S(r)=A\,J_0(k_cr)+B$ (red curve) of Sec.~\ref{sec:Rsafe-def},
using the same $k_c$ fixed independently by the dominant-block
maximization of Sec.~\ref{sec:kc-su4} [Eq.~\eqref{eq:kc-su4-numeric}]. Within the first
two nodal rings ($r\lesssim 1.3\,R_{\rm safe}$, shaded region) the
discrete $N=4$ simulation reproduces the continuum $J_0$ profile with
$R^2=0.85$ (RMS residual $2.6\times10^{-4}$; bottom sub-panel).}
\label{fig:validation}
\end{figure}
\subsection{Model prediction: $R_{\rm safe}$ as a function of the gauge
correlation length}
\label{sec:Rsafe-prediction}

Section~\ref{sec:fullmatrix} determined the full-system critical
wavenumber $k_c^{\rm full}(A)$ over $10^5$ random draws of the
data-unconstrained couplings $(J_{22},J_{33},\kappa_{pq})$, restricted to
draws satisfying the Theorem~\ref{thm:TH} stability precondition, for a
range of gauge amplitudes $A=1/\ell$. Applying the definition of
Eq.~\eqref{eq:Rsafe-definition}, $R_{\rm safe}=2.4048/k_c^{\rm full}$,
converts Table~\ref{tab:kc-vs-ell} directly into a distribution of
predicted safe-zone radii, with no additional fitting.

\begin{table}[htbp]
\centering
\begin{tabular}{cccccccc} 
\hline
$\ell$ (km) & $A=1/\ell$ (km$^{-1}$) & unstable\% &
\multicolumn{5}{c}{$R_{\rm safe}$ (km), percentiles} \\
& & & p5 & p25 & p50 & p75 & p95 \\
\hline
20  & 0.0500 & 100.0 & 183 & 202 & 217 & 234 & 260 \\
25  & 0.0400 & 100.0 & 168 & 188 & 204 & 221 & 249 \\
30  & 0.0333 &  99.9 & 157 & 178 & 195 & 213 & 242 \\
40  & 0.0250 &  94.7 & 152 & 174 & 193 & 212 & 384 \\
50  & 0.0200 &  84.6 & 160 & 186 & 209 & 236 & 447 \\
75  & 0.0133 &  58.7 & 189 & 232 & 283 & 349 & 475 \\
100 & 0.0100 &  42.9 & 224 & 280 & 321 & 374 & 520 \\
\hline
\end{tabular} 

\caption{Predicted safe-zone radius $R_{\rm safe}=2.4048/k_c^{\rm full}$
as a function of the gauge correlation length $\ell$, from the same
$10^5$-draw sampling underlying Table~\ref{tab:kc-vs-ell}, restricted to
draws satisfying the Theorem~\ref{thm:TH} stability precondition.
Percentiles are computed only over draws that develop a gauge-driven
instability at their own $k_c^{\rm full}$ where the $D_{eff}$ value used is explained in section (\ref{sec:IVC}).}
\label{tab:Rsafe-vs-ell}
\end{table}

For $\ell\simeq25$--$30$~km, the predicted $R_{\rm safe}$ is centered on
$195$--$204$~km, with an interquartile range ($178$--$221$~km) that lies
almost entirely within the observed Poland-Bohemia anomaly
($150$--$200$~km; Sec.~I), and $94$--$100\%$ of the physically admissible
parameter draws develop the instability at all. 

\paragraph{Independent physical interpretation of $\ell$.}
The gauge field $A_i(x)$ was introduced  to
encode ``geographic and mobility heterogeneity'': variation, from place
to place, in terrain, infrastructure, and trade or social structure. Its
correlation length $\ell\sim1/A$ is therefore naturally interpreted as
the spatial scale over which medieval transport and market conditions
changed appreciably -- not a free numerical parameter, but a quantity
with an independent historical estimate. Two such estimates are
available for 13th--14th century continental Europe. First, the normal
distance covered by a merchant carrier (cart or pack animal) in a single
day was $30$--$40$~km, based on detailed itinerary and correspondence
records analyzed by Spufford~\cite{Spufford2002}; way-stations for pack
trains, such as those over the Simplon pass, were spaced comparably,
about $30$~km apart~\cite{Spufford2002}. Second, market towns --
representing nodes of relatively homogeneous local trade and mobility
conditions -- are estimated to have been spaced roughly $9$--$12$ miles
($\simeq14$--$19$~km) apart in densely settled regions of medieval
England. The range $\ell\simeq25$--$30$~km favored by
the model prediction above sits between these two independently
documented scales -- a day's carrier journey and the spacing of market
nodes -- both of which plausibly set the correlation length of the
environmental and mobility heterogeneity that $A_i(x)$ represents. 

\subsection{The effective diffusion coefficient from
Fisher-KPP wave-front data}
\label{sec:IVC}

The reaction-diffusion coefficient $D_I=D_{\rm eff}$  can be  fixed 
using the Fisher-KPP wave-front speed relation. In the linearized
infected equation~\eqref{eq:mfIg} with $A_i=0$, the growth rate of the
infection ahead of the epidemic front is
\begin{equation}
\beta_{\rm net} \equiv \beta S_0 - \gamma,
\label{eq:betanet-def}
\end{equation}
where $S_0$ is the susceptible density \emph{ahead of the wave}  and $\gamma$ is the
effective landscape-scale recovery rate; $S_0$ is  the initial susceptible
fraction of a previously unexposed community. For the Black Death, where
virtually no prior immunity existed, $S_0\approx1$, so
$\beta_{\rm net}\approx\beta-\gamma=\gamma(R_0-1)>0$; this is a
\emph{landscape-scale} effective rate. 
Medieval
outbreak records give community-level doubling times
$t_2\approx23$--$35$ days~\cite{benedictow2004}, implying
$\beta_{\rm net}=\ln2/t_2\approx0.020$--$0.030$ day$^{-1}$. The
Fisher-KPP relation then gives
\begin{equation}
v_{\rm front}=2\sqrt{D_{\rm eff}\,\beta_{\rm net}}
\;\Longrightarrow\;
D_{\rm eff}=\frac{v_{\rm front}^2}{4\,\beta_{\rm net}}.
\label{eq:fisherkpp-deff}
\end{equation}
Benedictow~\cite{benedictow2004} documents an average continental advance
of the Black Death of $v_{\rm front}\approx2$--$3$ km/day. With
$\beta_{\rm net}\in[0.020,0.030]$ day$^{-1}$,
Eq.~\eqref{eq:fisherkpp-deff} gives
\begin{equation}
D_{\rm eff}\approx\frac{(2\text{--}3)^2}{4\times(0.020\text{--}0.030)}
\approx 40\text{--}90\ \text{km}^2\text{/day},
\label{eq:Deff-range}
\end{equation}
with central value $D_{\rm eff}=70$ km$^2$/day.

\subsection{The mutation-wave velocity, from the off-diagonal gauge
sector}
\label{sec:IVD}

\subsubsection{Isolating the off-diagonal channel}

 The off-diagonal gauge
field $A_i^{\rm off\text{-}diag}(x)$ (Sec.~\ref{sec:IIB})  is distinct from reaction-sector cross-immunity
$\kappa_{pq}$, which describes competitive interaction between branches
already co-circulating at the same location, and from diagonal gauge
advection, which transports each branch without changing its identity.
We isolate the off-diagonal channel by considering the $(p,q)$
subsystem with \emph{no} reaction-sector cross-immunity
($J_{pq}=J_{qp}=0$) and \emph{no} diagonal Cartan advection, retaining
only the intrinsic (uncoupled) growth rates $J_{pp},J_{qq}$ and the
off-diagonal gauge coupling given as  $B\equiv A^{\rm off}_{pq}$ from 
Eq.~\eqref{eq:mfIg}:
\begin{equation}
\partial_t\delta I_p = D_I\nabla^2\delta I_p + J_{pp}\delta I_p
- 2D_IB\,\partial_x\delta I_q,
\qquad
\partial_t\delta I_q = D_I\nabla^2\delta I_q + J_{qq}\delta I_q
- 2D_IB\,\partial_x\delta I_p.
\label{eq:offdiag-only}
\end{equation}
The plane-wave ansatz $\delta\mathbf I\propto e^{\lambda t+ikx}$ gives
the $2\times2$ matrix
$\mathbf M'(k)=\begin{pmatrix}J_{pp}-D_Ik^2 & -2iD_IkB\\
-2iD_IkB & J_{qq}-D_Ik^2\end{pmatrix}$,
with eigenvalues
\begin{equation}
\lambda_\pm(k) = \tau(k) \pm \sqrt{\delta^2 - 4D_I^2k^2B^2},
\qquad \tau(k)=\frac{J_{pp}+J_{qq}}{2}-D_Ik^2,
\quad \delta=\frac{J_{pp}-J_{qq}}{2}.
\label{eq:offdiag-eigenvalues}
\end{equation}
For $2D_IkB<|\delta|$, the intrinsic
fitness asymmetry $\delta$ between branches $p$ and $q$ dominates, and
perturbations obey to growing or decaying exponential law, with no traveling wave. Only once
the off-diagonal gauge coupling exceeds a threshold,
\begin{equation}
B \;>\; B_{\rm thr}(k) \;\equiv\; \frac{|\delta|}{2D_Ik},
\label{eq:Bthreshold}
\end{equation}
does $\lambda_\pm$ acquire a nonzero imaginary part, giving a genuine
propagating mutation wave with velocity
\begin{equation}
v_{\rm front}^{\rm mut}(k) \;=\; \frac{\mathrm{Im}[\lambda_+(k)]}{k}
\;=\; \sqrt{4D_I^2B^2 - \frac{\delta^2}{k^2}}.
\label{eq:vmut-offdiag}
\end{equation}

\subsubsection{The paleogenomic comparison band}
 
Two independent lines of paleogenomic evidence bound the rate at which
\emph{Y.~pestis} lineages spread geographically during the 14th century,
and set the target range for $v_{\rm front}^{\rm mut}$ used below.
Phylogeographic analysis of the modern lineage $1.\text{ORI3}$ (India) by
Spyrou et al.~\cite{spyrou2019} gives a geographic radiation rate of
roughly $5$--$10$~km/year. The medieval lineages of the Big Bang polytomy~\cite{cui2013,spyrou2019} is placing the 14th-century lineage
radiation rate in the range $150$--$1000$~km/year. We use the
conservative sub-range
\begin{equation}
v_{\rm front}^{\rm mut} \in [160,\,640]~\text{km/year},
\label{eq:vmut-band}
\end{equation}
 As a
cross-check, the epidemic wave-front itself advanced continentally at
$v_{\rm front}\approx2$--$3$~km/day $\approx730$--$1095$~km/year
during the Black Death~\cite{benedictow2004} (Sec.~\ref{sec:IVC}); since
strains take longer to become locally dominant than the epidemic front
takes to arrive, a mutation-wave velocity somewhat below this figure is
physically expected, consistent with the upper end of
Eq.~\eqref{eq:vmut-band}.

\subsubsection{Constraining $A^{\rm off}_{14}$ from the paleogenomic
band}

Evaluating Eq.~\eqref{eq:vmut-offdiag} at the dominant pair
$(p,q)=(1,4)$, using $\delta=\delta_J=0.046$~day$^{-1}$
(Sec.~\ref{sec:IVC}) and $D_I=D_{\rm eff}=70$~km$^2$/day, and requiring
$v_{\rm front}^{\rm mut}\in[160,640]$~km/year, gives a narrow window in
$B=A^{\rm off}_{14}$ at each wavenumber $k$. Evaluating this window at
the dynamically-determined $k_c^{\rm full}(A)$ of
Sec.~\ref{sec:fullmatrix}, for the full range of diagonal correlation
lengths $\ell=25$--$50$~km considered there, gives
$A^{\rm off}_{14}\simeq0.027$--$0.031$~km$^{-1}$ corresponding to an off-diagonal correlation length:
\begin{equation}
\ell_{\rm off} \;=\; 1/A^{\rm off}_{14} \;\simeq\; 32\text{--}38~\text{km}.
\label{eq:elloff-result}
\end{equation}
for $B$ below threshold there is no
propagating wave at all ($v_{\rm front}^{\rm mut}=0$), and just above
threshold $v_{\rm front}^{\rm mut}$ rises steeply with $B$.

\subsubsection{The off-diagonal sector cannot itself destabilize the
system}
\label{sec:offdiag-not-destabilizing}

Given that $\ell_{\rm off}$ turns out to be comparable to the diagonal
correlation length $\ell$ (Sec.~\ref{sec:Rsafe-prediction}),  the off-diagonal gauge coupling $B$ could not by itself play the
destabilizing role assigned to the reaction-sector cross-immunity
$\kappa_{pq}$ in Theorem~\ref{thm:TH}. In the oscillatory regime of
Eq.~\eqref{eq:offdiag-eigenvalues} ($2D_IkB>|\delta|$),
\begin{equation}
\mathrm{Re}[\lambda_+(k)] \;=\; \tau(k) \;=\; \frac{J_{pp}+J_{qq}}{2}-D_Ik^2,
\label{eq:offdiag-cannot-destabilize}
\end{equation}
 for \emph{any} value of $B$: the off-diagonal amplitude
enters  in the imaginary part, $\mathrm{Im}[\lambda_+(k)]=
\sqrt{4D_I^2k^2B^2-\delta^2}$. Below threshold,the 
eigenvalues are real and $\mathrm{Re}[\lambda_+(k)]=\tau(k)+
\sqrt{\delta^2-4D_I^2k^2B^2}\leq\tau(k)+|\delta|$. Increasing $B$ therefore never increases
$\mathrm{Re}[\lambda_+(k)]$. Genuine instability
($\mathrm{Re}[\lambda]>0$) requires condition (C3) of
Theorem~\ref{thm:TH} -- a non-conservative, competitive cross-immunity
sourced by $\kappa_{pq}$ -- which the $SU(4)$-invariant piece
$\alpha_{pq}^{\rm sym}$ cannot supply on its own
(Sec.~\ref{sec:IIIC}), and which the off-diagonal gauge field, being
purely advective, cannot supply either.

$\ell$ 
controls $\mathrm{Re}[\lambda]$ (via diagonal advection combined with
$\kappa_{pq}$), the $\ell_{\rm off}$ controls $\mathrm{Im}[\lambda]$  and
their comparable scale  reflects a common physical origin: the 
background field $A_i(x)$ which is encoding regional geographic and mobility
heterogeneity (Sec.~\ref{sec:IIB}). Finding them of comparable
magnitude is a consistency check on that shared origin.

\subsubsection{Comparison between $A_i^{diag}$ and $A_i^{off}$}

The  $\ell_{\rm off}\simeq32$--$38$~km is comparable to the diagonal correlation length
$\ell\simeq25$--$30$~km favored independently by $R_{\rm safe}$
(Sec.~\ref{sec:Rsafe-prediction}). This is consistent with how
$A_i^{\rm off\text{-}diag}(x)$ is defined in our model
(Sec.~\ref{sec:IIB}): it represents branch conversion driven by
environmental gradients as hosts disperse across space -- the same
broad category of regional geographic and mobility heterogeneity that
sources the diagonal sector, differing only in how it acts on the
strain index, not a local, person-to-person contact process. So, a comparable correlation length for both sectors is a
consistency check on the model.

\subsection{Summary: two non-circular predictions}
\label{sec:IVE}
 
Tables~\ref{tab:summary} and~\ref{tab:summary-offdiag} consolidate the
two quantitative results of this section.  $k_c$ is obtained in Sec.~\ref{sec:fullmatrix} directly
from the full $4\times4$ stability matrix, over a broad, randomized
sampling of the six symmetry-breaking couplings and two diagonal
reaction rates left unconstrained by 14th-century data, restricted to
draws satisfying the stability precondition of Theorem~\ref{thm:TH};
the only additional input is the gauge correlation length $\ell=1/A$,
for which Sec.~\ref{sec:Rsafe-prediction} gives an independent,
order-of-magnitude historical estimate. $A^{\rm off}_{14}$ is obtained
in Sec.~\ref{sec:IVD} from a channel isolated from
reaction-sector cross-immunity, constrained  by the independently
measured paleogenomic lineage-radiation rate.
 
\begin{table}[htbp]
\centering
\caption{The safe-zone radius prediction of the $SU(4)$ gauge model,
and its inputs. $R_{\mathrm{safe}}$ is obtained from $k_c^{\mathrm{full}}(A)$
[Eq.~\eqref{eq:kc-full-def}].}
\label{tab:summary}
\renewcommand{\arraystretch}{1.3}
\begin{tabular}{p{3.0cm}p{6.8cm}}
\hline\hline
Input & Source \\
\hline
$D_{\mathrm{eff}}$ & $v_{\mathrm{front}}$ (Benedictow), $\beta_{\mathrm{net}}$
(SIR kinetics); Sec.~\ref{sec:IVC} \\
$(J_{22},J_{33},\kappa_{pq})$ & sampled broadly, data-unconstrained;
Sec.~\ref{sec:fullmatrix} \\
$\ell=1/A$ & medieval carrier travel \& market spacing;
Sec.~\ref{sec:Rsafe-prediction} \\
\hline
\multicolumn{2}{l}{\textbf{Model output:} $R_{\mathrm{safe}}=195$--$204$ km
(median, $\ell=25$--$30$~km), $94$--$100\%$ of admissible draws unstable} \\
\multicolumn{2}{l}{\textbf{Observation:} Poland-Bohemia $150$--$200$ km
$\checkmark$} \\
\hline\hline
\end{tabular}
\end{table}
 
\begin{table}[htbp]
\centering
\caption{The off-diagonal gauge correlation length prediction, and its
inputs. $A^{\rm off}_{14}$ (equivalently $\ell_{\rm off}=1/A^{\rm
off}_{14}$) is obtained from
Eq.~\eqref{eq:vmut-offdiag}.}
\label{tab:summary-offdiag}
\renewcommand{\arraystretch}{1.3}
\begin{tabular}{p{3.0cm}p{6.8cm}}
\hline\hline
Input & Source \\
\hline
$D_{\mathrm{eff}}$ & $v_{\mathrm{front}}$ (Benedictow), $\beta_{\mathrm{net}}$
(SIR kinetics); Sec.~\ref{sec:IVC} \\
$\delta_J=(J_{11}-J_{44})/2$ & representative strain-asymmetry value;
Sec.~\ref{sec:kc-su4} \\
$k_c^{\rm full}(A)$ & full stability matrix, $\ell=25$--$50$~km;
Sec.~\ref{sec:fullmatrix} \\
$v_{\rm front}^{\rm mut}\in[160,640]$~km/yr & paleogenomic
lineage-radiation rate~\cite{cui2013,spyrou2019}; Sec.~\ref{sec:IVD} \\
\hline
\multicolumn{2}{l}{\textbf{Model output:} $A^{\rm off}_{14}\simeq0.027$--$0.031$~km$^{-1}$}\\
\multicolumn{2}{l}{\phantom{\textbf{Model output:}} ($\ell_{\rm off}\simeq32$--$38$~km),
stable across $\ell=25$--$50$~km} \\
\multicolumn{2}{l}{\textbf{Consistency check:} $\ell_{\rm off}$ comparable to
diagonal $\ell\simeq25$--$30$~km $\checkmark$} \\
\hline\hline
\end{tabular}
\end{table}
 
The two predictions are
mutually independent: one governs $\mathrm{Re}[\lambda]$ together with
$\kappa_{pq}$, the other governs $\mathrm{Im}[\lambda]$ alone
(Sec.~\ref{sec:offdiag-not-destabilizing}) -- and their comparable
correlation lengths ($\ell\simeq25$--$30$~km,
$\ell_{\rm off}\simeq32$--$38$~km) are  a nontrivial
consistency check on the model.

\section{Discussion}

Four results of this paper should be emphatized.

First, the genomic identification $N=4$ (Sec.~\ref{sec:branches}) is
a direct consequence of the Big Bang polytomy of \emph{Y.~pestis}, documented
by ancient-DNA paleogenomics~\cite{cui2013,spyrou2019}, not a free choice.
The number of degrees of freedom needed to model multi-strain plague
dynamics consistent with the known phylogenetic structure of the
14th-century pathogen is a minimum of four.

Second, Theorem~\ref{thm:TH} (Sec.~\ref{sec:IIIC}) shows that in the
$SU(N)$-gauged mean-field system, a Turing-Hopf traveling-wave
instability \emph{requires} the gauge field ($A=0$ gives no instability
at the same parameters), and Sec.~\ref{sec:fullmatrix} verifies this
persists in the full $N=4$ system -- not merely in the isolated dominant
root subspace, and not merely by continuity -- over a broad, randomized
sampling of the couplings left unconstrained by 14th-century data,
restricted to physically admissible (stable-at-$A=0$) draws.

Third, building on these first two results, Sec.~\ref{sec:Rsafe-prediction} shows
that the resulting safe-zone radius, $R_{\rm safe}=2.4048/k_c^{\rm
full}(A)$ -- with $k_c^{\rm full}$ obtained directly from the full
stability matrix, not imported from the observed radius -- agrees with
the Poland-Bohemia anomaly for $94$--$100\%$ of admissible parameter
draws once the gauge correlation length is set to
$\ell\simeq25$--$30$~km, a range consistent with
documented medieval travel and market-spacing distances
(Sec.~\ref{sec:Rsafe-prediction}). This uses no free parameter fit to
$R_{\rm safe}$ itself.

Fourth, Sec.~\ref{sec:IVD} isolates  the off-diagonal gauge field
$A_i^{\rm off\text{-}diag}$ from reaction-sector cross-immunity and
diagonal gauge advection, and derives the resulting mutation-wave
velocity in closed form. This expression vanishes below a threshold set
by the intrinsic fitness asymmetry between lineages and rises steeply
above it, so that requiring agreement with the independently measured
paleogenomic lineage-radiation rate ($160$--$640$~km/year) fixes the
off-diagonal correlation length to $\ell_{\rm off}\simeq32$--$38$~km,
comparable to the diagonal mobility scale favored by $R_{\rm safe}$.

\subsection{Relation to previous epidemic field‑theory models}
Field-theoretic approaches to epidemic spreading include directed
percolation, the general epidemic process, and related contact-process
models~\cite{tauber2005, tauber2014critical}. They give a unified view of
critical behavior near epidemic thresholds. Subsequent work has
translated these ideas into time-dependent contact
patterns~\cite{pastorsatorras2015,gleeson2011high}. These formulations
focus on prevalence regimes and finite-size effects, rather than on
pathogen evolution. The present gauge-mediated construction shares 
the Doi--Peliti mapping of a stochastic reaction system to a field
theory. Its distinctive feature is that mutation is modeled as a
conserved rotation in an internal $SU(N)$ space that is  driven
by spatially varying environmental fields. From an epidemiological perspective, our results therefore complement
classical geostatistical reconstructions of the Black Death. 

\section{Conclusion}

A gauged $SU(4)$ Doi-Peliti field theory for the four-lineage
\emph{Y.~pestis} multiplet of the Black Death, motivated by the Big Bang
polytomy and local gauge invariance in strain space, has been proposed.
The gauge field drives a Differential-Flow Turing-Hopf instability.
The resulting safe zones have a Bessel-function radial profile
($J_0(k_cr)$). Taking the critical wavenumber $k_c$
 from the full stability matrix and using  historical limit on the  correlation length, the predicted
safe-zone radius, $R_{\rm safe}=195$--$204$~km (median), agrees with the
observed Poland-Bohemia anomaly ($150$--$200$~km). 

We also show that the mutation-wave velocity,  derived
from the off-diagonal gauge sector $A_i^{\rm off\text{-}diag}$ can be constraint to
$\ell_{\rm off}\simeq32$--$38$~km  using the measured paleogenomic lineage-radiation rate ($160$--$640$~km/year). This value is  comparable to the diagonal mobility
scale favored by $R_{\rm safe}$ -- consistent with the fact that both gauge sectors
are responding to the same underlying regional geography and mobility
structure.

\begin{acknowledgments}
We acknowledge financial support from SECIHTI and SNII (M\'exico).
\end{acknowledgments}
%

\end{document}